\documentclass[sigconf,natbib=false]{acmart}
\AtBeginDocument{%
  }

\copyrightyear{2026}
\acmYear{2026}
\setcopyright{cc}
\setcctype{by}
\acmConference[CODASPY '26]{Proceedings of the Sixteenth ACM Conference on Data and Application Security and Privacy}{June 23--25, 2026}{Frankfurt am Main, Germany}
\acmBooktitle{Proceedings of the Sixteenth ACM Conference on Data and Application Security and Privacy (CODASPY '26), June 23--25, 2026, Frankfurt am Main, Germany}
\acmDOI{10.1145/3800506.3803498}
\acmISBN{979-8-4007-2562-3/2026/06}

\RequirePackage[
  datamodel=acmdatamodel,
  style=acmnumeric,
  ]{biblatex}

\newcommand{\R}{\ensuremath{\mathcal{R}}}
\newcommand{\LL}{\ensuremath{\mathcal{L}}}
\newcommand{\NIZK}{\ensuremath{\mathsf{NIZK}}}

\newcommand{\HE}{\ensuremath{\mathsf{HE}}}
\newcommand{\RR}{\ensuremath{\mathcal{R}_\mathsf{range}}}
\newcommand{\KC}{\textbf{KC}}

\newcommand{\Setup}{\ensuremath{\mathsf{Setup}}}
\newcommand{\Prove}{\ensuremath{\mathsf{Prove}}}
\newcommand{\Verify}{\ensuremath{\mathsf{Verify}}}
\newcommand{\Com}{\ensuremath{\mathsf{Com}}}
\newcommand{\KeyGen}{\ensuremath{\mathsf{KeyGen}}}
\newcommand{\Commit}{\ensuremath{\mathsf{Commit}}}
\newcommand{\Enc}{\ensuremath{\mathsf{Enc}}}
\newcommand{\Dec}{\ensuremath{\mathsf{Dec}}}
\newcommand{\Eval}{\ensuremath{\mathsf{Eval}}}
\newcommand{\Open}{\ensuremath{\mathsf{Open}}}
\newcommand{\Sign}{\ensuremath{\mathsf{Sign}}}

\newcommand{\pp}{\ensuremath{\mathsf{pp}}}
\newcommand{\ppc}{\ensuremath{\mathsf{ppc}}}
\newcommand{\pk}{\ensuremath{\mathsf{pk}}}
\newcommand{\vk}{\ensuremath{\mathsf{vk}}}
\newcommand{\sk}{\ensuremath{\mathsf{sk}}}
\newcommand{\evk}{\ensuremath{\mathsf{evk}}}
\newcommand{\rlk}{\ensuremath{\mathsf{rlk}}}
\newcommand{\bsk}{\ensuremath{\mathsf{bsk}}}

\newcommand{\st}{\mathit{st}}
\newcommand{\chall}{\mathit{chall}}

\newcommand{\IsValid}{\ensuremath{\mathsf{IsValid}}}

\newcommand{\Exp}{\ensuremath{\mathbf{Exp}}}
\newcommand{\Adv}{\ensuremath{\mathbf{Adv}}}
\newcommand{\oRegister}{\textrm{\sc Register}}
\newcommand{\oUpdate}{\textrm{\sc Update}}
\newcommand{\oAttest}{\textrm{\sc Attest}}
\newcommand{\oVerifyNew}{\textrm{\sc VerifyNew}}
\newcommand{\oCorrupt}{\textrm{\sc Corrupt}}
\newcommand{\oChallenge}{\textrm{\sc Challenge}}
\newcommand{\A}{\mathcal{A}}
\newcommand{\B}{\mathcal{B}}
\newcommand{\C}{\mathcal{C}}
\newcommand{\Cr}{\ensuremath{\mathit{Cr}}}
\newcommand{\cnt}{\ensuremath{\mathsf{cnt}}}
\newcommand{\LLL}{\ensuremath{\mathit{L}}}
\newcommand{\res}{\ensuremath{\mathsf{res}}}
\newcommand{\msg}{\ensuremath{\mathsf{msg}}}

\newcommand{\game}{\ensuremath{\mathbf{Game}}}

\newcommand{\getsr}{{\:{\leftarrow{\hspace*{-3pt}\raisebox{.75pt}{$\scriptscriptstyle\$$}}}\:}}

\usepackage{amsmath, array, amsthm}
\usepackage{float}
\usepackage{graphicx}
\usepackage{tabularray}

\usepackage[ruled, vlined, linesnumbered]{algorithm2e}

\let\oldnl\nl
\newcommand{\nonl}{\renewcommand{\nl}{\let\nl\oldnl}}

\SetKwComment{Comment}{$\triangleright$ }{}

\usepackage[many,skins]{tcolorbox}
\usepackage{url}
\usepackage[shortlabels,inline]{enumitem}
\usepackage{multirow, multicol}
\usepackage{balance}
\usepackage{circuitikz}
\usepackage{caption}
\usepackage{subcaption}

\setlist[itemize]{nosep, leftmargin=*, topsep=0pt, partopsep=0pt, itemsep=1pt, parsep=0pt}
\setlist[enumerate]{nosep, leftmargin=*, topsep=0pt, partopsep=0pt, itemsep=1pt, parsep=0pt}
\setlist[description]{nosep, leftmargin=*, topsep=0pt, partopsep=0pt, itemsep=1pt, parsep=0pt}

\usepackage[nameinlink]{cleveref}
\crefformat{figure}{#2(figure~#1)#3}
\Crefformat{figure}{#2Figure~#1#3}
\crefformat{table}{#2(table~#1)#3}
\Crefformat{table}{#2Table~#1#3}

\usepackage{rotating}
\usepackage{makecell}
\usepackage{pgfplots}
\pgfplotsset{compat=1.18}
\usepackage{tikzit}
\usepackage{mathtools}
\usepackage{breqn}
\usepackage{xfrac}

\DeclareMathAlphabet{\mathsc}{T1}{cmr}{m}{sc}
\newtheorem{definition}{Definition}

\newcounter{claimcounter}
\renewcommand{\theclaimcounter}{\arabic{claimcounter}}

\renewenvironment{proof}{{\par\noindent\textit{Proof Sketch.}}}{\hfill\qed\par}

\usepackage{varwidth}
\usepackage{adjustbox}
\usepackage{color, colortbl}
\usepackage{multirow}
\definecolor{anti-flashwhite}{rgb}{0.95, 0.95, 0.96}
\definecolor{Gray}{gray}{0.9}
\definecolor{airforceblue}{rgb}{0.36, 0.54, 0.66}
\definecolor{azure}{rgb}{0.023,0.603,0.952}
\definecolor{skygrayblue}{rgb}{0.878,0.952,1}
\definecolor{aliceblue}{rgb}{0.94, 0.97, 1.0}
\definecolor{alizarin}{rgb}{0.82, 0.1, 0.26}
\definecolor{amber}{rgb}{1.0, 0.75, 0.0}
\definecolor{amber(sae/ece)}{rgb}{1.0, 0.49, 0.0}
\definecolor{bronze}{rgb}{0.8, 0.5, 0.2}
\definecolor{bananayellow}{rgb}{1.0, 0.88, 0.21}
\definecolor{battleshipgrey}{rgb}{0.52, 0.52, 0.51}
\definecolor{bole}{rgb}{0.47, 0.27, 0.23}
\definecolor{bulgarianrose}{rgb}{0.28, 0.02, 0.03}
\definecolor{cadet}{rgb}{0.33, 0.41, 0.47}
\definecolor{ceil}{rgb}{0.57, 0.63, 0.81}
\definecolor{cerulean}{rgb}{0.0, 0.48, 0.65}
\definecolor{charcoal}{rgb}{0.21, 0.27, 0.31}
\definecolor{coolblack}{rgb}{0.0, 0.18, 0.39}
\definecolor{coolgrey}{rgb}{0.55, 0.57, 0.67}
\definecolor{darkcandyapplered}{rgb}{0.64, 0.0, 0.0}
\definecolor{darkbrown}{rgb}{0.4, 0.26, 0.13}
\definecolor{darkcerulean}{rgb}{0.03, 0.27, 0.49}
\definecolor{darkgray}{rgb}{0.66, 0.66, 0.66}
\definecolor{darkjunglegreen}{rgb}{0.1, 0.14, 0.13}
\definecolor{darktaupe}{rgb}{0.28, 0.24, 0.2}
\definecolor{davysgrey}{rgb}{0.33, 0.33, 0.33}
\definecolor{frenchblue}{rgb}{0.0, 0.45, 0.73}
\definecolor{almond}{rgb}{0.94, 0.87, 0.8}
\definecolor{beaublue}{rgb}{0.74, 0.83, 0.9}
\definecolor{beige}{rgb}{0.96, 0.96, 0.86}
\definecolor{bisque}{rgb}{1.0, 0.89, 0.77}
\definecolor{black}{rgb}{0.0, 0.0, 0.0}
\definecolor{forestgreen}{rgb}{0.2, 0.6, 0.2}
\definecolor{fluorescentorange}{rgb}{1.0, 0.75, 0.0}
\definecolor{ghostwhite}{rgb}{0.97, 0.97, 1.0}
\definecolor{antiquewhite}{rgb}{0.98, 0.92, 0.84}
\definecolor{LightCyan}{rgb}{0.88,1,1}

\newcommand{\CSP}{\textbf{SP}}

\UseTblrLibrary{diagbox}

\newcommand{\protocol}{\texttt{H-Elo}}

\begin{document}


\title[Hidden Elo]{Hidden Elo: Private Matchmaking through Encrypted Rating Systems}


\author{Mindaugas Budzys}
\orcid{0000-0002-1913-7985}
\affiliation{%
  \institution{Tampere University}
  \city{Tampere}
  \country{Finland}}
\email{mindaugas.budzys@tuni.fi}

\author{Bin Liu}
\orcid{0000-0002-6591-3711}
\affiliation{%
  \institution{Tampere University}
  \city{Tampere}
  \country{Finland}}
\email{bin.liu@tuni.fi}
\additionalaffiliation{%
  \institution{grchain.io}
  \city{Shenzhen}
  \country{China}
}

\author{Antonis Michalas}
\orcid{0000-0002-0189-3520}
\affiliation{%
  \institution{Tampere University}
  \city{Tampere}
  \country{Finland}}
\email{antonios.michalas@tuni.fi}



\begin{abstract}
Matchmaking has become a prevalent part in contemporary applications, being used in dating apps, social media, online games, contact tracing and in various other use-cases. However, most implementations of matchmaking require the collection of sensitive/personal data for proper functionality. As such, with this work we aim to reduce the privacy leakage inherent in matchmaking applications. 
We propose \protocol{}, a Fully Homomorphic Encryption (FHE)-based, private rating system, which allows for secure matchmaking through the use of traditional rating systems. In this work, we provide the construction of \protocol{}, analyse the security of it against a capable adversary as well as benchmark our construction in a chess-based rating update scenario. Through our experiments we show that \protocol{} can achieve similar accuracy to a plaintext implementation, while keeping rating values private and secure. Additionally, we compare our work to other private matchmaking solutions as well as cover some future directions in the field of private matchmaking. To the best of our knowledge we provide one of the first private and secure rating system-based matchmaking protocols.

\end{abstract}



\maketitle

\section{Introduction}
\label{sec:intro}

In the current digital era, a wide range of applications and communications have moved to online spaces. As such, a larger emphasis is made on making professional or personal connections through online services. Notably, making connections has become easier due to matchmaking services available on social media, dating apps or in online games. One of the ways, to implement matchmaking is by using a rating system, such as the Elo rating system~\cite{glickman1995comprehensive}. In this case, players are matched with each other depending on their respective ratings. This form of matchmaking with rating systems have been used in a variety of zero-sum games, such as chess, traditional sports and e-sports. Additionally, the Elo system has been used in ecological studies~\cite{porschmann2010male} and artificial intelligence evaluation~\cite{askell2021general}. In most applications, the rating is calculated by a server and revealed publicly to all parties.

While this approach is convenient, it introduces some notable drawbacks and privacy concerns in matchmaking. Namely, in some cases the ratings are hidden from users and only the server has access to it, leaving users in the dark of their rating, such as in the case of the previous Tinder algorithm~\cite{soemardi2023testing}. Additionally, rating systems have drawn criticism for being easy to exploit for selective pairing~\cite{vevcek2014comparison}, where a player can aim to search for opponents where they maximize gains and minimize losses, especially when players can choose their own opponents. Alongside this, revealing the exact rating of a person can introduce ways for malicious parties to monitor specific user ratings and identify patterns in their behaviour. Due to these reasons, there is a clear need for securing rating calculations, allowing people to know their exact rating while being able to fairly match with opponents and avoid prying eyes.

To address these issues and encourage further research into private matchmaking through rating systems, we propose \texttt{Hidden Elo} (\protocol{}) -- a Fully Homomorphic Encryption (FHE)-based protocol for rating updates and verification. More specifically, we make use of the Cheon-Kim-Kim-Song Residual Number System (CKKS-RNS)~\cite{cheon2018full} FHE scheme (we explain our rationale for choosing CKKS-RNS in \autoref{ssec:design-rat}) as well as Zero-Knowledge Proofs (ZKPs) to implement a matchmaking protocol, which hides a user's rating from other parties, while updating the rating \textit{correctly} through a central service provider. We show that \protocol{} allows for private rating-based matchmaking with reasonable accuracy by implementing and testing the rating update of our approach.

\smallskip

\noindent \textbf{Contributions.} \enskip The main contributions of this paper can be summarized as follows:

\begin{enumerate}[\bfseries C1.,nosep]
    \item We propose \protocol{}, one of the first private and secure rating system-based matchmaking protocols. We achieve this by making use of FHE to allow blind calculations of a rating update after a number of matches. Our construction uses lattice-based cryptography and therefore achieves \textit{post-quantum} security. Additionally, our construction leverages ZKPs to reveal the minimum information needed to allow proper matchmaking, while keeping the actual ratings private.  
    \item We provide an extensive security analysis, which proves, that \protocol{} holds security notions for providing fair matchmaking as well as keeping the true rating hidden from a semi-honest adversary.
    \item We implement the FHE rating update of \protocol{} and run benchmarks on the implementation. Through our evaluation we analyse the runtime and memory costs of \protocol{} as well as compare the accuracy of the rating update to plaintext computations. We show that our approach has reasonable runtime results for users using standard HE security parameters. Additionally, we show that our approach provides negligible deviation from the exact rating calculated on plaintext even after~10,000 consecutive updates. 
    \item Finally, we study existing works in the area of private matchmaking and identify that only Matchmaking Encryption (ME)~\cite{ateniese2021match} can be compared, in a meaningful and fair way, to this work. To this end, we provide the first systematic comparison between ME and an FHE-based, post-quantum-secure matchmaking protocol, highlighting differences in trust assumptions, security guarantees, and performance.
\end{enumerate}

\paragraph{Organization}
The rest of the paper is organized as follows. In \Cref{sec:app-domain} we cover the possible usecases of our approach covering the areas where \protocol{} can be used. 
In \Cref{sec:rel-work} we cover the related work in the field, covering some of the rating systems that exist as well as private and secure matchmaking solutions currently available in literature. In \Cref{sec:preliminaries} we introduce the cryptographic primitives which we make use of in our construction, and in \Cref{sec:protocol-cons} we provide the detailed construction of \protocol{}. Afterwards, \Cref{sec:security} covers the threat model and security analysis of \protocol{}. We implement our benchmarking usecase and cover the results in \Cref{sec:eval-results}. Finally, we provide our conclusions and cover future work in \Cref{sec:conclusions}.

\section{Application Domain}
\label{sec:app-domain}

Matchmaking is a core system in the current digital world, allowing two or more users to match with each other through common interests in social media or by having similar skill levels in online games. However, as these systems evolve and become more complex the inner workings of them become inaccessible or incomprehensible to users, all the while requiring large amounts of personal data or user activity to operate. As such, it is natural for people to want to protect sharing their private information with matchmaking services while retaining the functionality of these services. 

Through \protocol{}, we aim to provide similar functionality as the plaintext implementation by matching users with similar skill levels based on rating systems, while keeping the rating secure on the client side. Consequently, such functionality would allow various real world usecases to be implemented in a more privacy-conscious way. These use-cases could include: \begin{enumerate*}[\itshape (i)] 
    \item implementing social media functionalities, such as forming professional or personal connections with people, without directly sharing their personal data to the service provider; 
    \item implementing online game matchmaking, where a user can know their true rating and have fair matchmaking by avoiding selective pairing;
    \item implementing reputations systems in marketplaces, where users and seller may vote on each others reputation without revealing the exact reputation but a generalization of how the seller or buyer is perceived. 
\end{enumerate*}
\section{Related Work}
\label{sec:rel-work}


\subsection{Elo and Traditional Rating Systems} 
\label{subsec:rw-mm}

In the 1950s, Arpad Elo developed the fundamental theory, which eventually became the U.S. Chess Federations (USCF) rating system. This version as well as a modification of it designed by Glickman~\cite{glickman1995comprehensive} is still used to this day in modern chess skill rating. This system adopted the name of the founder and is known as the "Elo rating system"~\cite{glickman1999rating}. The goal of the system is to calculate the likelihood of a player winning against their opponent and compare it to the actual result of the game. If a player with a lower rating wins against an opponent with a higher one, then the player "beat the odds" and gains a larger increase in their rating. On the other hand, the loss of rating is smaller if the likelihood to win was lower.


In addition to the Elo rating system, other rating systems have also been adopted throughout history, such as Glicko~\cite{glickman2012example}, which introduce a variable for uncertainty in the player's skill. More specifically, uncertainty shows the possible deviation of the player's skill, where a player with high uncertainty will lose or gain rating faster than those with low uncertainty.

To the best of our knowledge, there are no works analysing secure utilisation of rating systems, however there are works analysing how to perform secure matchmaking through Functional Encryption (FE) constructions, such as Attribute-based Encryption (ABE) and 
ME, or through Private Set Intersection (PSI) solutions.

\subsection{Functional Encryption (FE)}
\label{subsec:rw-fe}

FE refers to public-key encryption schemes, which allows a specific function to be computed on a ciphertext through the use of a decryption key. This term was generalized by Boneh et al.~\cite{boneh2011functional} and ever since then, the field has expanded to different function with the most notable ones being computing dot products with Inner-Product FE (IPFE)~\cite{abdalla2020inner} and computing quadratic equations through Quadratic FE (QFE)~\cite{agrawal2021multi}. More recently, there has been a push by researchers to develop FE constructions that expand the functionality of FE to allow qualitative analysis of encrypted data~\cite{nuoskala2024spade} with selective, partial decryption. 


\subsubsection{Attribute-based Encryption (ABE)}
One type of FE constructions focused on fine-grained access control on who would be able to decrypt a ciphertext. This initially started as Identity-Based Encryption (IBE)~\cite{shamir1985identity}, which was later expanded to Ciphertext-Policy ABE (CP-ABE) by Sahai and Waters~\cite{sahai2005fuzzy}. CP-ABE allowed identifying users not just through a string of characters like in IBE, but through a set of descriptive attributes. In turn CP-ABE associated these sets of attributes to a generated decryption key which was then able to decrypt a ciphertext that had an equivalent access policy associated with it.

Later on, researchers constructed Key-Policy ABE (KP-ABE)~\cite{goyal2006attribute} following an inverse construction than CP-ABE, namely where the decryption keys are associated with the access policy, while the ciphertexts are associated with the encrypters attributes. In both KP-ABE and CP-ABE power is given to a single entity (users in CP-ABE or the trusted authority in KP-ABE) who defines who can access what data. Additionally, in both methods a mismatch between the attributes and the policy reveals them to the querier~\cite{attrapadung2015duality}. Due to this, Ateniese et al. proposed Matchmaking Encryption (ME)~\cite{ateniese2021match}.



\subsubsection{Matchmaking Encryption} 
ME is a special case of ABE, where both sender and receiver (who have individual attributes) can set their own required policies for a message to be revealed. More specifically, ME allows for both the sender and the receiver to set their own policies $\mathbb{S}$ and $\mathbb{R}$, respectively. Next, the sender associates his own attributes $\sigma$ as well as policy $\mathbb{S}$ to the ciphertext $c$. Then, a receiver with attributes $\rho$ and policy $\mathbb{R}$ are matched against $\mathbb{S}$ and $\sigma$ in $c$, if all attributes match the policies, then the receiver can decrypt $c$, otherwise an error occurs and no one learns anything aside from the fact that there was a attribute mismatch.

ME gives the ability for both sender and receiver to define their attributes and their own policy. Through this, a ciphertext can be decrypted only when a match occurs and both user attributes match each others policy. Additionally, ME addresses the issue of attribute/policy hiding and prevents this information leak when an attribute mismatch occurs during decryption. 

While ME is a highly important step for providing a solution for secure matchmaking, it has noticeable drawbacks within it. Some of these include the strong assumptions required in constructing ME, making it impractical to implement. While these assumptions can be lowered to a weaker version, they results in a lowered security guarantee which can only support a bounded number of queries. This results in ME not being applicable to situations where users would wish to arbitrarily match with \textit{multiple users.}

\subsection{Private Set Intersection (PSI)}

PSI is a cryptographic method that enables parties to compute the intersection of their private input sets without revealing any information beyond the intersection~\cite{pinkas2018scalable}. 
In a PSI protocol, two or more parties can compute the intersection of their private input sets while ensuring that the elements not in the intersection remain private.
PSI has a wide range of applications, including contact tracing~\cite{wu2023efficient} and
matchmaking applications~\cite{he2022differentially}.

PSI has several variants, each defined by the inputs, outputs, and interactions among the parties involved. Each variant has specific requirements to consider when designing and implementing protocols. For instance, in traditional PSI, one or both parties learn the intersection of their input sets~\cite{pinkas2018scalable}. In contrast, some protocols may compute a \textit{function} based on the intersection without revealing the intersection itself~\cite{garimella2021private}. Additionally, specific protocols may require prior authorization for elements in the sets or impose a threshold on the intersection size before producing the output~\cite{liu2023efficient}. 


Several cryptographic techniques are commonly employed to enhance the efficiency and security of PSI protocols. Examples of cryptographic techniques used include \textit{Hashing}~\cite{pinkas2018efficient}, \textit{Oblivious Transfer}~\cite{rindal2017improved}, \textit{Homomorphic Encryption}~\cite{cong2021labeled}, Garbled Circuits~\cite{pinkas2018scalable} and Oblivious Pseudo-Random Functions~\cite{kolesnikov2016efficient}.
\textit{Hashing} improves efficiency by mapping set items into bins, significantly reducing computational overhead, though it requires careful parameter selection to avoid collisions. \textit{Oblivious Transfer (OT)} enables secure two-party computation, allowing a receiver to obtain one of many data pieces without the sender knowing which was chosen. However, it may require multiple communication rounds~\cite{rindal2017improved}. 
HE allows computations on encrypted data without decryption, enabling complex secure operations at the cost of high computational overhead~\cite{cong2021labeled}. \textit{Garbled Circuits} facilitate secure function evaluation on private inputs but can be less efficient than specialized PSI protocols~\cite{pinkas2018scalable}. Finally, \textit{Oblivious Pseudo-Random Functions (OPRFs)} securely generate pseudo-random values for PSI, balancing efficiency and security while requiring careful design to resist attacks.

PSI protocols have evolved since their initial introduction. Performance has been a central focus, as PSI protocols are known to be computationally intensive. Despite progress in performance, these protocols remain complex and require careful instantiation, considering platform-specific characteristics like computation limits, bandwidth, or network delay.

\section{Preliminaries}
\label{sec:preliminaries}

\subsection{Elo Rating System} 

The system aims to calculate the likelihood of a player winning against their opponent and compare it to the result of the game.
Traditionally, a new rating in the Elo system is calculated as such:
\begin{equation}
    R'=R_{prev}+K\cdot(S_{real}-S_{exp})
    \label{eq:elo-rating}
\end{equation}
where $R_{prev}$ and $R'$ refers to the player's previous and new rating, respectively; $K$ is a K-factor, used for linear score adjustments; $S_{real}$ and $S_{exp}$ are the real and expected outcome of the match for the player. $S_{exp}$ is calculated with the following equation:
\begin{equation}
    S_{exp}=\frac{1}{1+10^{(R_{opp}-R_{player})/400}}
    \label{eq:exp-result}
\end{equation}
where $R_{opp}$ and $R_{player}$ are the Elo ratings of the opponent and player, respectively.

\subsection{Homomorphic Encryption}
\label{subsec:he}

HE relies on the property of homomorphism, more specifically ``privacy homomorphism''. This property allows computation on encrypted data and was initially theorized by Rivest, Adleman, and Dertouzos~\cite{rivest1978data} for the RSA cryptosystem to allow performing arithmetic multiplications on ciphertexts. This type of encryption was further expanded into Fully HE (FHE) by Craig Gentry in~2009~\cite{gentry2009fully} and allowed an unbound number of multiplications and additions on a ciphertext. However, due to the costs imposed by the bootstrapping algorithm, researched turned towards limiting the number operations to avoid computing the bootstrapping step. This led to Levelled HE (LHE) schemes, such as BGV~\cite{brakerski2014leveled}, B/FV~\cite{fan2012somewhat}, CKKS~\cite{cheon2017homomorphic} and some FHE schemes, e.g. TFHE~\cite{chillotti2020tfhe}, to be proposed. The field continued to be developed and allowed LHE schemes to be expanded into FHE schemes, such as in the case of CKKS-RNS~\cite{cheon2018full}, which allowed bootstrapping on the CKKS scheme. 
While there are differences between HE schemes, all of them rely on the same core algorithms, which can be defined as:

\begin{definition}[Homomorphic Encryption]
\label{def:he}
Let $\mathsf{HE}$ for message space $\mathcal{M}$, ciphertext space $\mathcal{C}$ and class of functions $\mathcal{F} : \mathcal{M}^n \to \mathcal{M},\; n\in\mathbb{N}^*$,\; be a homomorphic encryption scheme with a quadruple of PPT algorithms $\mathsf{HE = (KeyGen, }$ $\mathsf{Enc, Eval, Dec)}$ such that:
\begin{itemize}
\item $\mathsf{HE.KeyGen(}1^{\lambda})$ -- a probabilistic algorithm, that takes as input the security parameter $\lambda$, and outputs a public key $\mathsf{pk}$, a private key $\mathsf{sk}$, and an evaluation key $\mathsf{evk}$.
\item $\mathsf{HE.Enc(pk}, x)$ -- a probabilistic algorithm, that takes as input the public key $\mathsf{pk}$ and a message $x\in\mathcal{M}$ and outputs a ciphertext $c\in\mathcal{C}$, where $x$ is encoded into a polynomial representation.
\item $\mathsf{HE.Eval(evk}, f, c_1, \dots, c_n)$ -- a probabilistic algorithm, that takes as input the evaluation key $\mathsf{evk}$, a function $f$, and a set of $n$ ciphertexts, and outputs a ciphertext $c_f$.
\item $\mathsf{HE.Dec}({\mathsf{sk}}, c)$ -- a deterministic algorithm, that takes as input the secret key $\mathsf{sk}$ and a ciphertext $c$, and outputs a plaintext $x$.
\end{itemize}

\underline{Correctness}: For all parameters $\lambda\in\mathbb{N},\, n\in\mathbb{N}^*$, $\left(\pk, \sk, \evk\right)\getsr$ $\HE.\KeyGen\left(1^\lambda\right)$, $f\in\mathcal{F}$ and ciphertexts $c_i \getsr \HE.\Enc\left(\pk,x_i\right)$,\; $x_i\in\mathcal{M}$,\; $i\in[n]$,
\[
\begin{aligned}
    \mathbb{P}[\HE.\Dec\left(\sk,\HE.\Eval\left(\mathsf{evk}, f, \left(c_1, \ldots, c_n\right)\right)\right) 
    \\
    \neq f\left(x_1, \ldots, x_n\right) + \delta] = negl\left(\lambda\right)
\end{aligned}
\]
where $\delta$ is a small value (or noise) that depends on the scaling factor $\Delta\in\mathbb{N}$. If $\HE$ is such that $\delta$ always vanishes, $\HE$ is said to be \textit{exact}. Otherwise, it is called \textit{approximate}.

\end{definition}

\subsubsection{Cheon-Kim-Kim-Song (CKKS) Scheme}

The CKKS scheme~\cite{cheon2017homomorphic}, proposed by Cheon \textit{et al.}, is the first HE scheme which enables encrypted arithmetic computations on approximate numbers. Through the use of CKKS a user can encrypt real numbers ($x\in\mathbb{R}$) and perform standard arithmetic, such as addition, subtraction and multiplication. Additionally, the construction allows for ciphertext rotations without necessitating decryption. The scheme was a major advancement from previous HE schemes, such as BGV~\cite{brakerski2014leveled} and B/FV~\cite{fan2012somewhat}, which allowed the same operations only on integer plaintexts. Similarly as the aforementioned integer HE schemes, CKKS is based on the quantum-resistant, lattice-based Ring Learning with Error (RLWE) problem providing security guarantees against quantum security threats. Also, with advances in CKKS bootstrapping~\cite{cheon2018full}, it allows high precision computations without having to rely on re-encryption to reduce noise while having lower latency than bootstrapping with the BGV scheme~\cite{al2023demystifying}.

Similarly to other HE schemes, CKKS ciphertexts accumulate noise after each arithmetic operations. As such, CKKS employs three techniques for noise management, namely modulus switching, scaling and relinearization. Modulus switching is a technique that changes the modulo of a ciphertext $ct$ from $q$ to a value of $q'<q$. This avoids quadratic noise blowup and keeps the noise level constant after multiplications~\cite{kim2021revisiting}. 
In addition to modulus switching, CKKS has to perform rescaling. Namely, as CKKS ciphertext are multiplied by a scale factor $\Delta$, this leads to multiplications of two ciphertext increasing the factor exponentially. As this results in modulus overflows and diminishes the precision, CKKS requires rescaling by dividing the ciphertext by $\frac{1}{\Delta}$ after each multiplication~\cite{agrawal2023high}. 
Lastly, relinearization is a technique to reduce the degree of a ciphertext, namely multiplications increases the degree of the ciphertext from degree-1 to degree-2, which increases the size of the ciphertext making future operations more expensive~\cite{kim2021revisiting}. Through relinearization this degree is once again reduced to a lower value by using a relinearization key $\mathsf{rlk}$.


Another important functionality of CKKS is batching~\cite{halevi2020design}, which allows encrypting multiple data points with a single encryption (i.e. encrypting an array of data points) and performing computations on each point at the same time. With this functionality CKKS can utilize Single Instruction, Multiple Data (SIMD) parallelism, which increases the speed as multiple operations can be done simultaneously. On the other hand, CKKS also supports sparse packing~\cite{cheon2018bootstrapping}, which reduces the amount of slots available to encode plaintext values into a ciphertext. While this does not provide benefits for reducing computation time as in the case of batching, however, it reduces the memory footprint of the cryptographic context.

\subsection{Range Proof with Commitment-Ciphertext Consistency}
\label{subsec:bg-zkp}
Zero-Knowledge Proofs (ZKPs) are cryptographic protocols, which allow a prover to convince a verifier that a statement is true without revealing any information beyond the validity of that statement. Since the seminal work of Goldwasser \textit{et al.} in ~\cite{goldwasser2019knowledge}, ZKPs have become a cornerstone of modern cryptography and are widely used in privacy-preserving authentication and identification~\cite{rosenberg2023zk}, machine learning~\cite{weng2021mystique}, and verifiable computation~\cite{ben2019aurora} protocols. 

Classical ZKPs are interactive: the prover and the verifier exchange multiple messages. In many practical settings, particularly those require public verifiability or verification in an asynchronous manner, such interaction is undesirable. This motivates the study of Non-Interactive Zero-Knowledge Arguments\footnote{When a proof system is only computationally sound, it is called an argument.} (NIZKs), in which a single proof is generated and can be verified without interaction. NIZKs are commonly obtained by applying the Fiat-Shamir heuristic~\cite{FS86} to an interactive public-coin proof system, or designed directly in the common reference string (CRS) model. In this work, we require the NIZKs to be knowledge sound (extractable) and therefore they are NIZK arguments of knowledge. For simplicity, we still use NIZK to denote such knowledge sound arguments. 


Let $\R \subseteq \{0,1\}^* \times \{0,1\}^*$ be an NP relation and let $\LL=\{x\ |\ \exists \ \omega: (x, \omega) \in \R\}$ be its corresponding language. 
Formally, a NIZK is defined as follows. 

\begin{definition}[Non-Interactive Zero-Knowledge Argument, NIZK]
A NIZK for a relation $\R$ consists of a triplet of PPT algorithms $\NIZK = (\Setup, \Prove, \Verify)$, where:

\begin{itemize}
    \item $\NIZK.\Setup(1^\lambda, \R) \rightarrow (\pk, \vk)$: takes as input the security parameter $\lambda$ and the NP relation $\R$, and outputs a proving key $\pk$ and a verifying key $\vk$. 
    \item $\NIZK.\Prove(\pk,x,\omega) \rightarrow \pi$: takes as input the proving key $\pk$, a statement $x$ and a witness $\omega$, and outputs a proof $\pi$.
    \item $\NIZK.\Verify(\vk,x,\pi) \rightarrow \{0,1\}$:  takes as input the verifying key $\vk$, a statement $x$ and a proof $\pi$, and outputs $1$ if the proof is valid for $x$ and $0$ otherwise. 
\end{itemize}
\label{def:zkp}
\end{definition}

Additionally, a NIZK argument system must satisfy the following properties: 
\begin{enumerate}
    \item \textbf{Completeness.} For all $(x,\omega) \in \R$, an honest prover using $\pk$ produces a proof $\pi$ such that the verifier using $\vk$ accepts with probability $1$.
    \item \textbf{Zero-knowledge.} There exists a PPT simulator, which given $(\pk,\vk)$ but without access to witnesses, can generate proofs that are computationally indistinguishable from real proofs produced with witnesses. 
    \item \textbf{Knowledge Soundness (extractability).} If any PPT adversary can produce a valid proof for a statement $x$ with non-negligible probability, then there exists an PPT extractor such that it is able to recover a valid witness for $x$ given access to the adversary except with negligible probability. 
\end{enumerate}

Range proofs are a specialized class of ZKPs in which a prover convinces a verifier that a secret value, which may be presented in either commitment-based or encryption-based form, lies within a prescribed interval without revealing any additional information about the value itself. Such proofs are fundamental to privacy-preserving financial applications, access control mechanisms and various consistency checks, where one needs to ensure the secret values respect some certain range. 

Although range proofs can be constructed in interactive form, in this work we only focus on the non-interactive setting, since our construction requires proofs that can be generated once and subsequently verified without any further interaction. Accordingly, throughout this paper we invoke the standard NIZK algorithms directly when employing range proofs in our protocols. Additionally, knowledge soundness guarantees that any prover that produces a valid proof must know a witness consistent with the relation. It is crucial for binding the underlying value of the ciphertext and the subsequent homomorphic computation over the ciphertext. 

In order to present the definition of range proofs, we rely on a commitment scheme $\Com=(\KeyGen,\Commit,\Open)$ that is both \textit{binding} and \textit{hiding}. 

\begin{definition}[Range Proof]
    Let $\ppc$ denote the commitment parameters and $r$ denote the randomness used for generating the commitment. Given integers $a,b,v \in \mathbb{Z}$ with $a\leq v<b$ and a commitment scheme $\Com=(\KeyGen,\Commit,\Open)$, we define the following NP relation:
    \[
    \begin{aligned}
        \RR = \{(x,\omega)\ |\ & x = (\ppc,C,a,b),\ \omega = (v,r), \\
        & C = \Commit_\ppc(v;r) \wedge a \leq v < b \}.
    \end{aligned}
    \]
    A (non-interactive) range proof is a NIZK argument for the relation $\RR$.
\end{definition}

In addition to prove that a committed Elo value lies in a certain interval, we still need to bind this value to the homomorphic ciphertext that the server will later update. Otherwise, a malicious user could present a valid range proof for some specified value while providing a ciphertext that encrypts a different value to bypass the rank restrictions. 

Following the commitment and proof of equivalence paradigm in \cite{LindellN18} 
but adapting it to our setting, we define the following NP relation $\R_{rccc}$:

\begin{definition}[Range Proof with Commitment-Ciphertext Consistency]
    Let $r,r'$ denote the randomness used for generating the commitment and encryption, respectively. Given integers $a,v,b \in \mathbb{Z}$ with $a\leq v<b$, a commitment scheme $\Com=(\KeyGen,\Commit,$ $\Open)$ and an homomorphic encryption scheme $\HE=(\KeyGen, \Enc, $ $\Eval, \Dec)$, a non-interactive Range proof with Commitment Ciphertext Consistency is a NIZK argument for the NP relation $\R_{rccc}$.
\[
    \begin{aligned}
    \R_{rccc} = \{ (x,\omega) \ |\ & x = (\ppc, \pk, C, c,  a, b),\ \omega = (v, r, r'), \\
    & C = \Commit_\ppc(v;r) \wedge a \leq v < b \wedge c = \Enc_\pk(v;r') \}.
\end{aligned}
\]
\end{definition}

Essentially, the NP relation defined above only rely on the encryption algorithm $\Enc$ and therefore could be instantiated with an arbitrary public key encryption scheme. However, in our setting the ciphertext is not merely an encryption of the committed value and it could be involved in a sequence of homomorphic updates carried out by the server. For this reason, we explicitly denote it as $\HE$, even though the evaluation algorithm $\Eval$ is not involved in the relation. This choice ensures the minimal requirement: it only checks the consistency of the commitment and the ciphertext. 

\section{\protocol~Construction}
\label{sec:protocol-cons}

In this section we present the construction of our FHE-based matchmaking protocol using the Elo rating system. We cover some additional notation we use in \autoref{app:notation}. Additionally, we provide a high-level overview and a visualisation of the protocol in \autoref{fig:high-level}. 

\begin{figure*}[t]
    \centering
    \includegraphics[width=.95\linewidth]{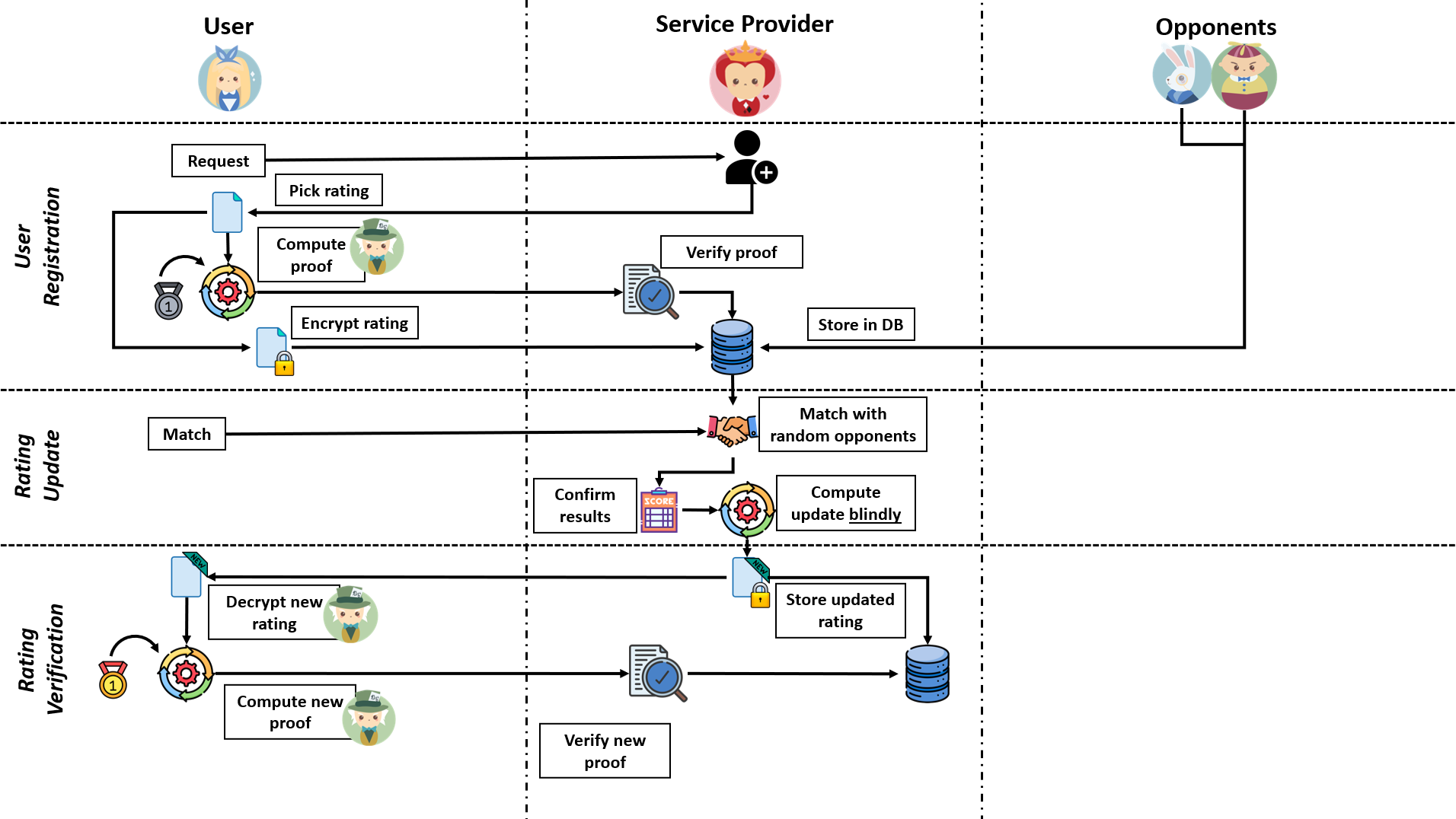}
    \caption{High-level Overview of \protocol{}}
    \label{fig:high-level}
\end{figure*}

\subsection{System Model}
\protocol{} consists of three entities: \begin{enumerate*}[\itshape (i)]
\item Service Provider (\CSP{}),
\item Clients ($\mathcal{U}$),
\item Trusted Key Curator (\KC{}).
\end{enumerate*}
\begin{itemize}
    \item Service Provider (\CSP{}) -- \CSP{} hosts the matchmaking service. They define the application and is responsible for initiating the match, verifying the outcome of a match and securely calculating the rating of the users after the match. Additionally, the \CSP{} provides the initial rating of each client and registers them to the system.
    \item Clients ($\{u_1,\ldots,u_n\} \in\mathcal{U}$) -- each client $u_i\in \mathcal{U}$ knows their rating and can prove their rank after a rating update. The clients wish to match with other similarly rated clients.
    \item Trusted Key Curator (\KC{}) -- \KC{} generates the public HE keys which are used by $\mathcal{U}$ in \protocol{}. In this case, \KC{} would decrypt and share the results with $\mathcal{U}$ after a rating update. The operations of the \KC{} can be executed on a Trusted Execution Environment (TEE).
\end{itemize}

\subsection{High-level Overview}

\protocol{} is executed in three parts: \begin{enumerate}
    \item \textit{\textbf{User Registration}} -- is comprised of \Cref{alg:init} and \Cref{alg:reg} and is run once by the parties. The $\mathbf{Init}$ algorithm generates the respective keys and public parameters used throughout \protocol. The $\mathbf{Register}$ algorithm allows $u_i$ to choose their initial rating, encrypt it with HE and compute the proof of their rank and sends the proof and ciphertext to \CSP. If the proof is valid, then $u_i$ is registered to the system and can match with other users $u_j$, otherwise the registration is aborted.
    \item \textit{\textbf{Rating Update}} -- is done through \Cref{alg:update} and \Cref{alg:announce}. After $u_i$ concludes $n$ matches, \CSP{} retrieves $u_i$'s encrypted rating as well as the ratings of his $n$ opponents and computes the new rating homomorphically following \Cref{eq:elo-rating} by running the $\mathbf{Update}$ algorithm. The phase end by the \KC{} revealing the updated rating to $u_i$ through the $\mathbf{Announce}$ algorithm.
    \item \textit{\textbf{Rating Verification}} -- is comprised of \Cref{alg:attest} and \Cref{alg:verify}. Upon receiving the new rating, $u_i$ computes a new commitment and ciphertext through the $\mathbf{VerifyNew}$ algorithm, which are sent to \KC. Afterwards, \KC{} runs the $\mathbf{Attest}$ algorithm to open the commitment and verify if the commitment is consistent with its' local record, then \KC{} attests the ciphertext and commitment with a signature. Finally, $u_i$ generates a new proof and sends the proof to the \CSP. If \CSP{} can verify the proof, then he updates the local record and allows $u_i$ to continue matchmaking with others.
\end{enumerate} 

\subsection{Construction}

\protocol{} is executed in three parts \begin{enumerate*}[\itshape (i)]
    \item \textit{User Registration},
    \item \textit{Rating Update},
    \item \textit{Rating Verification}.
\end{enumerate*} 
The \protocol{} protocol is constructed with six algorithms $\mathbf{\protocol =}$ $\mathbf{(Init, Register, Update, Announce, Attest,}$ $ \mathbf{VerifyNew)}$, which can be defined as follows: 
\begin{definition}[\protocol{}]
Let $\mathcal{L}$ denote a predefined set of ranks, where each rank $l_i \in \mathcal{L}$ have a defined minimum $l_i^{min} \in \mathbb{Z}$ and maximum $l_i^{max}\in \mathbb{Z}$ values required for the associated rank. Let $\mathsf{HE}$ define a homomorphic encryption scheme and $\lambda$ define the security parameter. Additionally, let $\mathsf{NIZK}$ define a non-interactive zero-knowledge proving system, $\mathsf{Com}$ a commitment scheme and $\Sigma$ a digital signature scheme.
Then \protocol{} consists of a 
tuple of PPT algorithms \protocol~$\mathbf{= (Init, Register, Update, Announce, Attest, VerifyNew)}$, where:
\begin{itemize}
    \item \protocol$\mathbf{.Init}(1^\lambda,\mathcal{R}_{rccc}) \rightarrow (\vk_{\mathbf{KC}}, \{pk_i\}_{i \in \mathcal{U}}, \ppc, \pk_{\mathbf{SP}}, \vk_{\mathbf{SP}})$ -- takes as input the security parameter $\lambda$ and an NP relation $\mathcal{R}_{rccc}$ and outputs \textbf{KC}'s verification key $\vk_{\mathbf{KC}}$, the public HE keys $pk_i$ of $u_i$, as well as the commitment public parameters $\ppc$ and proving $\pk_{\mathbf{SP}}$ and verification keys $\vk_{\mathbf{SP}}$ used for $\mathsf{NIZK}$ range proofs.
    \item \protocol$\mathbf{.Register}(R_{init},l_{init},\mathsf{pk}_{\mathbf{SP}}) \rightarrow \pi_{i_0}$ -- takes as input the initial rating $R_{init} \in \mathbb{Z}$, and initial rank $l_{init}$ and outputs the initial proof $\pi_{i_0}$ and registers $u_i$ to the system.
    \item \protocol$\mathbf{.Update}(c_i,c^{j}_{opp},n,S_{real},K) \rightarrow c'_i$ -- takes as input $u_i$'s encrypted rating $c_i$, the number of matches $n>2$, the encrypted ratings of $u_i$'s opponents $c^{j}_{opp}$, where $j\in[0,\ldots,n]$, summed match outcome $S_{real}$, and K-factor $K$ and outputs $u_i$'s encrypted updated rating $c'_i$.
    \item \protocol$\mathbf{.Announce}(c_i, \sk_i) \rightarrow R'_i$ -- takes as input the $u_i$'s encrypted rating $c_i$, and the secret key $\sk_i$ and output $u_i$'s plaintext Elo value $R'_i$.  
    \item \protocol$\mathbf{.Attest}(c'_i,\mathsf{sk}_{\mathbf{KC}}) \rightarrow \sigma$  -- takes as input $u_i$'s encrypted new rating $c'_i$ and the \textbf{KC}'s signing key $\mathsf{sk}_{\mathbf{KC}}$ and outputs a signature $\sigma$ on $u_i$'s commitment.
    \item \protocol$\mathbf{.VerifyNew}(id_i,R'_i,C'_i,r, c''_i, \sigma) \rightarrow \pi_i$ -- takes as input $u_i$'s $id_i$, new rating $R'_i$, the commitment $C'_i$ on $R'_i$, the randomness $r$ used to generate $C'_i$, current encrypted rating $c''_i$, the signature $\sigma$ and outputs the proof $\pi_i$ of $u_i$'s rank.  
\end{itemize}
\label{def:helo}
\end{definition}

\paragraph{\textbf{Algorithms.}} The \textit{User Registration} phase involves generating the required keys for running \protocol{} and registering a new user to the system. In \Cref{alg:init}, \KC{} generates $\vk_{\mathbf{KC}}$ and $\sk_{\mathbf{KC}}$ for $\Sigma$ (line 1). Then \KC{} generates $\pk_i$, $\sk_i$ and $\mathsf{evk}_i$ used for $\mathsf{HE}$ for each user $u_i\in\mathcal{U}$ (lines 2-3). \CSP{} generates the public parameters $\mathsf{pp}$ used in $\mathsf{Com}$ as well as $\pk_{\mathbf{SP}}$ and $\vk_{\mathbf{SP}}$ keys linked to $\R_{rccc}$ used in $\mathsf{NIZK}$ (lines 4-5). All the public keys are shared to other entities (line 6). 
In \Cref{alg:reg}, a new user is registered into the matchmaking system by \CSP{} defining an ID $id_i$ for $u_i$ and $u_i$ obfuscates his rating $R_i$ by some noise $\alpha$ (lines 2-3) and generates the  commitment $C_i$ (line 4), ciphertext $c_i$ (line 5) and proves that rating is within set bounds by the expected rank $l_{init}$, which is shared to \CSP{} (lines 6-9). The \CSP{} verifies the proof and publishes the user $id_i$ and rank $l_i$ (lines 10-14). If the verification fails  \CSP{} aborts the registration (lines 15-16). Once $u_i$ is registered, he may begin using the matchmaking service and matching with $u_j$, where $j \neq i$. After each match $\mathsf{cnt}_i$ is incremented by one.

\begin{algorithm}[ht!]
\caption{$\mathsf{\mathbf{\protocol.Init}}$}\label{alg:init}
\small
        \KwIn{Security parameter $1^\lambda$ and NP relation $\R_{rccc}$.} 
        \KwOut{\KC{}'s verification key $\vk_{\mathbf{KC}}$ of digital signature scheme, public parameters of the commitment scheme $\ppc$, proving key $\pk_{\mathbf{SP}}$ and verification key $\vk_{\mathbf{SP}}$ for the range proof. }
        \nonl\DontPrintSemicolon \KC{} generates the key pair for the digital signature scheme as well as the Homomorphic Encryption scheme. \;
        \PrintSemicolon $(\vk_{\mathbf{KC}}, \sk_{\mathbf{KC}}) \gets \Sigma.\KeyGen(1^\lambda)$\;
        \For{$u_i \in \mathcal{U}$}{
        $(\pk_i, \sk_i,\evk_i) \gets \mathsf{HE}.\KeyGen(1^\lambda)$\;
        }
        \nonl\DontPrintSemicolon \textbf{SP} runs the setup algorithm of the commitment scheme and the range proof. \;
        $\ppc \gets \Com.\KeyGen(1^\lambda)$\;
        $(\pk_{\mathbf{SP}}, \vk_{\mathbf{SP}}) \gets \NIZK.\Setup(1^\lambda, \R_{rccc})$\;
        \Return $(\vk_{\mathbf{KC}}, \{pk_i\}_{i \in \mathcal{U}}, \ppc, \pk_{\mathbf{SP}}, \vk_{\mathbf{SP}})$
\end{algorithm}

\begin{algorithm}[ht!]
\caption{$\mathsf{\mathbf{\protocol.Register}}$}\label{alg:reg}
    \small
    \KwIn{Initial rating $R_{init}$, initial rank $l_{init}$, 
    public parameters 
    $\ppc$, user $u_i$'s public key $\pk_i$, 
    \textbf{SP} proving key $\mathsf{pk}_{\mathbf{SP}}$.}
    \KwOut{Initial proof for $u_i$'s $\pi_{i_0}$.}
    \nonl\DontPrintSemicolon \textbf{SP} assigns ID's to the users who wish to register and sends them $R_{init}$ and their initial rank $l_{init}$.\;
    \PrintSemicolon$id_i \gets \mathsf{Rand(\mathbb{Z}^+)}$\;
    \nonl\DontPrintSemicolon The user adds some noise $\alpha$ to the initial rating and proves that it stays inside the range of the initial rank $l_{init}$.\;
    \PrintSemicolon $\alpha \gets \mathsf{Rand(\mathbb{Z})}$\;
    $R_{i} = R_{init} + \alpha$\;
    $C_{i} \gets \Com.\Commit_\ppc(R_i;r)$\;
    $c_{i} \gets \mathsf{HE.Enc(pk}_i, R_{i};r')$\;
    $x_{i_0} \gets (\ppc, \pk_i, C_i, c_i, l^{min}_{init}, l^{max}_{init})$\;
    $\omega_{i_0} \gets (R_{i}, r, r')$\;
    $\pi_{i_0} \gets \NIZK.\Prove(\pk_{\mathbf{SP}}, x_{i_0}, \omega_{i_0})$\;
    \Return $(id_i, x_{i_0}, \omega_{i_0}, \pi_{i_0})$\;
    \nonl\DontPrintSemicolon On receiving $(x_{i_0}, \pi_{i_0})$ from $u_i$, \textbf{SP} verifies the proof and publishes the user and the rank if the proof is valid.\;
    \PrintSemicolon$b \gets \NIZK.\Verify(\vk_{\mathbf{SP}}, x_{i_0}, \pi_{i_0})$\;
    \eIf{$b == 1$}
        {
        $\mathsf{cnt}_i = 1$\;
        \textbf{SP} records $(id_i, l_i, c_i, \mathsf{cnt}_i)$\;
        \Return $b$ and publish $id_i$ and $l_i$.}
        {\Return $\perp$ and abort the registration.}
\end{algorithm}

The \textit{Rating Update} phase consists of the $\mathbf{Update}$ algorithm, that runs when $\mathsf{cnt}_i == n$. The steps for updating $u_i$'s rating are defined in~\Cref{alg:update}. In this case \CSP{} evaluates the expected result of each match up to $j==n$ and creates a sum of all expected results $c_{exp-sum}$ (lines 1-7) as defined in \Cref{eq:exp-result}. Afterwards the update process follows \Cref{eq:elo-rating}. It is important to note, that \CSP{} evaluates the bootstrapping circuit on the resulting ciphertext after subtracting $c_{exp-sum}$ from $S_{real}$ to refresh the accumulated noise to an acceptable level (lines 8-9). The bootstrapping step is done here due to the plaintext values being the lowest at this point of the protocol, allowing the noise to be refreshed while retaining high precision. After the updated rating $c'_i$ (lines 10-11) is computed it is sent to \KC{} to be announced to $u_i$ as shown in~\Cref{alg:announce}.

\begin{algorithm}[ht!]
\caption{$\mathsf{\mathbf{\protocol.Update}}$}\label{alg:update}
    \small
        \KwIn{$u_i$'s encrypted current rating $c_{i}$, opponents encrypted ratings $c^{j}_{opp}$, number of matches before an update $n$, summed outcome of the matches $S_{real}$, K-factor $K$.}
        \KwOut{$u_i$'s encrypted updated rating $c'_{i}$.}
        \nonl \DontPrintSemicolon The server homomorphically calculates the expected match results $c_{exp-sum}$:\;
        \PrintSemicolon $c_{1} \gets \mathsf{HE.EvalMul}(c_{i},0.0025)$\Comment*[r]{$\frac{1}{400}=0.0025$} 
        $c_{exp-sum} = 0$\;
        \For{$j < n$}{
        $c_{2} \gets \mathsf{HE.EvalMul}(c^{j}_{opp},0.0025)$\; 
        $c_{3} \gets \mathsf{HE.EvalSub}(c_2,c_1)$\Comment*[r]{$c_3 = (c_2 - c_1)\div 400$}
        $c_{exp} \gets \mathsf{HE.EvalPoly}(c_{3},\frac{1}{1+10^{c_3}})$\;
        $c_{exp-sum} \gets \mathsf{HE.EvalAdd}(c_{exp-sum},c_{exp})$\;
        }
        \nonl \DontPrintSemicolon The server homomorphically calculates the updated rating and sends it back to the player.\;
        \PrintSemicolon$c_{4} \gets \mathsf{HE.EvalSub}(S_{real},c_{exp-sum})$\; 
        $c'_{4} \gets \mathsf{HE.EvalBootstrap}(c_{4})$\;
        $c_{5} \gets \mathsf{HE.EvalMul}(K,c_{4})$\;
        $c'_{i} \gets \mathsf{HE.EvalAdd}(c_{i},c_{5})$\Comment*[r]{$c'_{i} = c_{i} + K \cdot (S_{real}-c_{exp-sum})$}
        \nonl \DontPrintSemicolon \Return $c'_{i}$
\end{algorithm}

The final phase in \protocol{} is the \textit{Rating Verification}, where $u_i$ receives his new ranking and must generate a new proof $\pi_i$ of their rank $R'_i$ and have it attested by \textbf{KC}. In \Cref{alg:attest}, $u_i$ generates a commitment $C_i$ of his rating $R_i$ (line 1), if \KC{} can open the commitment, then \KC{} attests the commitment and sends the signature to $u_i$, otherwise the attestation fails (lines 2-5). In \Cref{alg:verify}, $u_i$ generates a ciphertext $c'_i$ of their rating and proves their rating with $\pi_i$ following the same steps as in \textit{User Registration} (lines 1-4). If \CSP{} can verify the proof, then the rank and ciphertext are updated, otherwise the verification is aborted (lines 5-10).


\begin{algorithm}[ht!]
\caption{$\mathsf{\mathbf{\protocol.Announce}}$}\label{alg:announce}
    \small
        \KwIn{$u_i$'s encrypted rating $c'_i$, \KC{}'s decryption key $\sk_i$.} 
        \KwOut{$u_i$'s rating $R_i$}
        \nonl\DontPrintSemicolon For every $n$ games, \CSP{} will send the latest encrypted rating $c'_i$ for $u_i$ to \KC{}. On receiving the encrypted contents, \KC{} decrypts $c'_i$ with the use of $\sk_i$ and sends it to $u_i$. \;
        \PrintSemicolon$R'_{i} \gets \mathsf{HE.Dec(sk}_i,c'_i)$\;
        \Return $R'_i$ and record 
        it for future attestation. 
\end{algorithm}

\begin{algorithm}[ht!]
\caption{$\mathsf{\mathbf{\protocol.Attest}}$}\label{alg:attest}
    \small
        \KwIn{$u_i$'s rating $R'_i$ and the encrypted rating obtained from \KC{} $c'_i$, \KC{}'s signing key $\sk_{\mathbf{KC}}$.} 
        \KwOut{A signature $\sigma$ on $u_i$'s commitment.}
        \nonl\DontPrintSemicolon The user $u_i$ computes the commitment $C_i$ and sends $(id_i, c'_i, C_{i})$ to \KC{} for attestation. \;
        $C_{i} \gets \Com_\ppc.\Commit(R'_{i};r)$\;
        \nonl\DontPrintSemicolon If $C_i$ can be opened as valid on $R'_i$, which is identical to \KC{}'s local record, and $R'_{i} \geq 0$, \KC{} attests $C_i$ by generating a signature $\sigma$ on $(id_i, c'_i, C_{i})$\;
        \eIf{$\Com_\ppc.\Open(C_{i}, R'_{i}, r) == 1$}{
            \Return $\sigma \gets \Sigma.\Sign(\sk_{\mathbf{KC}}, (id_i, c'_i, C_{i}))$.
            }
        {
            \Return $\perp$ and abort the attestation.
        }
\end{algorithm}

\begin{algorithm}[ht!]
\caption{$\mathsf{\mathbf{\protocol.VerifyNew}}$}\label{alg:verify}
    \small
        \KwIn{$u_i$'s id $id_i$, new rating $R'_{i}$, a commitment $C_i$ on $R'_{i}$ with opening randomness $r$, current encrypted rating $c''_i$ for $R'_i$, and a signature $\sigma$
        } 
        \KwOut{$u_i$'s proof of rank $\pi_{i}$.}
        \nonl\DontPrintSemicolon $u_i$ finds out the expected rank $l_{exp}$ for $R'_i$, then generates the range proof.\;
        $c'_{i} \gets \mathsf{HE.Enc(pk}_i, R'_{i};r')$\;
        $x_{i} \gets (\ppc, \pk_i, C_i, c'_i, l^{min}_{exp}, l^{max}_{exp})$\;
        $\omega_{i} \gets (R'_{i}, r, r')$\;
        $\pi_{i} \gets \NIZK.\Prove(\pk_{\mathbf{SP}}, x_{i}, \omega_{i})$\;
        \nonl\DontPrintSemicolon $u_i$ sends $(c''_i, c'_i, C_i, \pi_i, \sigma)$ to \textbf{SP}. The \textbf{SP} verifies the attestation and the proof. In the case that both are valid and $c''_i$ is identical to the local record $c_i$, \CSP{} then updates $c_i$ and $l_i$ on record.\;
        \PrintSemicolon $b \gets \NIZK.\Verify(\vk_{\mathbf{SP}}, x_{i}, \pi_{i}) \wedge \Sigma.\Verify(\vk_{\mathbf{KC}}, \sigma, (id_i, c''_i, C_i)) \wedge (c_i == c''_i) \wedge (\cnt_i == N)$\;
        \eIf{$b == 1$}{
            $c_{i} \gets c'_{i}$, $l_i \gets l_{exp}$ and $\mathsf{cnt_i} \gets 1$\;
            \Return $(c_i, l_i)$
            }
        {
            \Return $\perp$ and abort the verification.
        }
\end{algorithm}

\section{Security}
\label{sec:security}
In this section, we formally present the security properties achieved by our construction $\protocol{}$. Notably, we stress that the security of the underlying cryptographic primitives, such as the $\mathsf{NIZK}$ proof system and the homomorphic encryption scheme $\mathsf{HE}$, is assumed to hold under their respective assumptions and is therefore out of the scope of this work. 

We begin by introducing two core security notions tailored to our setting:
\begin{enumerate}
    \item \textbf{N-round Fair Matchmaking.} Each user's rating is correctly and fairly updated for every $N$ rounds. It ensures that users are always matched within the correct rank interval; and
    \item \textbf{Hidden Elo.} The service provider learns no more about a user's exact Elo value than what is implied by the validity of the range proofs. 
\end{enumerate}



In earlier sections, we presented \protocol{} core algorithms in minimal form for simplicity, however we expand our presentation in which local states are explicit and algorithm interfaces are associated with the local states. We stress that these changes are regarding presentation rather than semantic. Concretely, a user $u_i$'s local state will become $\st_{u_i} = (id_i, l_i, R_i, c_i, $ $ \sigma)$, where $id_i$ is its identity, $l_i$ is its current rank, $R_i$ is the Elo value, $c_i$ is the encrypted Elo value and $\sigma$ is the latest attestation obtained from \KC{}. For simplicity we consider the randomness used for generating the encrypted Elo value is a part of $R_i$, but when it comes to value of $R_i$ we just omit those randomness. Similarly, we also consider the commitment $C_i$ as a part of $c_i$. Additionally, \KC{}'s local state $\st_{KC} = (\sk_{KC}, \{\st_i\}_{u_i \in \mathcal{U}}, M)$, where $M$ is a list (indexed by users) recording each user's latest Elo value in plaintext form for future attestation. \CSP{}'s local state contains a list $\LLL = \{(id_i, l_i, c_i, \cnt_i)\}_{u_i \in \mathcal{U}}$ of per-user entries to record user's current rank and the encrypted Elo value. The counter $\cnt_i$ is used for tracking the number of matches since the last rank update. Some of the protocol algorithms are then understood as sub-protocols that take the relevant local states as inputs and update the corresponding states. This modification makes the security definitions more general, since they are no longer hard-wire our concrete H-Elo construction. Our adversary model also becomes strictly stronger in an explicit way. Because upon an adaptive takeover of a user, the adversary should now learn about the session related information (e.g. the randomness for commitment and encryption and the most recent attestation).

We first formalise the security notion of \textbf{N-round Fair Matchmaking}. In a system where user ranks are refreshed every N games, rank updates should correctly reflect the user's current Elo value. Under a semi-honest \CSP{}, this rank correctness condition suffices to guarantee fair matchmaking, i.e., users are paired with opponents of the same rank whenever matchmaking is invoked between rank refresh events. Let $\LL$ denote the predefined rank set, which is assumed to be valid in the sense that every admissible Elo value lies in exactly one rank interval. Let $\IsValid$ be a predicate used by the experiment to reject invalid inputs, with $\IsValid(R_{init}, l_{init})=1$ indicates the initial Elo value and the rank are valid and $\IsValid(S_{real})$ = 1 indicates that the summed match outcome is valid. 


The security property is captured via the following game-based experiment in which an adversary attempts to make the server record an incorrect rank for some user while the challenger executes the protocol. The challenger runs $\mathsf{Init}$ with the security parameter $1^\lambda$ and $\R_{rccc}$, which is the range proof relation that enforces commitment-ciphertext consistency. It also initialises three lists. The list $\LLL$ and $M$ mirror those in \CSP{}'s and \KC{}'s local states, but are made explicit in the experiment for clarity. $\Cr$ records the corrupted users. The challenger then hands the public parameters $\pp$ to the adversary $\A$, who then interacts with the oracles shown in \autoref{fig:foracles} 

The registration oracle $\oRegister(i, \st^*_u, R_{init}, l_{init})$ registers user $u_i$ with the initial Elo $R_{init}$ and rank $l_{init}$. If $u_i$ is corrupted, the oracle runs $\mathsf{Register}$ using the local state specified by the adversary $\st^*_u$. Otherwise it uses the local state of $u_i$ which is maintained by the experiment instead. Upon successful run, $\LLL$ is updated accordingly and the handle and returns $u_i$ to indicate that. If $u_i$ is corrupted, the oracle also releases the user's updated local state to $\A$. 

The update oracle $\oUpdate(i, j, S_{real})$ processes the released sum of match result $S_{real}$ for the user $u_i$ and $u_j$. If $\IsValid(S_{real})=1$, each user's game counter $\cnt$ is incremented by $1$ and $\mathsf{Update}$ is run on their encrypted Elo values. Whenever either user reaches the $N$-th game, the oracle invokes $\mathsf{Announce}$ to let \KC{} reveal the latest Elo value to that user. The plaintext Elo value will be recorded in $M$ for future attestation and the counter is reset to $1$. Because both the \KC{} and \CSP{} are semi-honest, their mutual interaction is reliable and outside the adversary's influence, the announcement is therefore modelled as part of $\oUpdate$. If either user is corrupt, the corresponding local-state update message is provided to $\A$.

The oracle $\oVerifyNew(i, \st^*_u)$ allows the adversary to request for a rank refresh of a specified user $u_i$. The oracle runs $\mathsf{VerifyNew}$ on $u_i$'s local state. If $u_i$ is corrupted, it uses the local state provided by the adversary instead. Since $\mathsf{VerifyNew}$ works only when the counter $\cnt_i == N$, $\cnt_i$ will be reset to $1$ on the successful run of the algorithm. After that, the $\oVerifyNew$ returns the encrypted Elo value and the updated rank. If $\A$ queries $\oCorrupt$ for taking over user $i$, $u_i$ will be added to the list $\Cr$ and the oracle returns the internal state of $u_i$. 

The adversary wins by specifying a user $u_i$ (and optionally a state $\st^*_u$). The challenger then executes $\mathsf{VerifyNew}$ using either the experiment's local state for $u_i$ or $\st^*_u$ if $u_i$ is corrupted. On the successful run of the algorithm, a rank $l_i$ can be derived from the outputs. If $l_i$ is different from the actual rank of $u_i$, computed by applying the deterministic function $\mathsf{rank}(\cdot)$ to the user's current Elo value maintained by the game (stored in $M[i]$), then the adversary is considered to win. 

\begin{definition}[N-round Fair Matchmaking]\label{def:fairness}
A private and secure rating system-based matchmaking protocol $\Pi$ is said to preserve $N$-round Fair Matchmaking for a valid predefined set of ranks $\LL$ with a set of user $\mathcal{U}$, for any probabilistic polynomial-time adversary $\A$, the following holds:
    \[
        \Adv^{\text{N-fairness}}_{\Pi, \A}(\lambda) = \Pr[\Exp^{\text{N-fairness}}_{\Pi, \A}(\lambda) \rightarrow 1]
    \]
  is a negligible function in $\lambda$, and the experiment $\Exp^{\text{N-fairness}}_{\Pi, \A}$ is defined as follows: 

\smallskip
  
  \begin{minipage}{\textwidth}
    \begin{small} \begin{tabbing}
        12\=12\=123\=12\=\kill
        \underline{$\mathbf{Exp}^{\text{N-fairness}}_{\Pi, \A}(\lambda)$} \\[3pt]
        \> $M, \LLL, \Cr \gets \bot$\\[1pt] 
        \> $(\st_{SP}, \st_{KC}, \{\st_{u_i}\}_{u \in \mathcal{U}}, \pp) \getsr \protocol{}.\mathsf{Init}(1^\lambda, \mathcal{R}_{rccc})$\\[1pt]
        \> $(\st^*_u, u_i) \getsr \A(1^\lambda, \pp: \mathcal{O}_{f})$ \\[1pt]
        \> If $u_i \not\in \mathcal{U}$ then return $\bot$ \\[1pt]
        \> If $u_i \in \Cr$ then $\st_{u_i} \gets \st^*_u$ \\[1pt]
        \> $\res \getsr \protocol{}.\mathsf{VerifyNew}(\st_{u_i},\st_{SP},\pp)$\\[1pt]
        \> If $\res \neq \bot$ then Parse $\res$ as $(c_i, l_i)$\\[1pt]
        \> If $l_i \neq \mathsf{rank}(M[i])$ then return $1$ \\[1pt]
        \>\> Else return $0$
      \end{tabbing}\end{small}
  \end{minipage}

  \begin{figure}[h]\scalebox{0.9}{
    \framebox{
      \hfill \begin{minipage}{.24\textwidth}
        \begin{small}\begin{tabbing}
            1\=1\=12\=12\=12\=\kill
            \>\underline{$\oRegister(i,\st^*_{u_i},R_{init},l_{init})$} \\[3pt]
            \>\> If $\LLL[i] \neq \emptyset \vee u_i \notin \mathcal{U}$ then return $\bot$\\[1pt]
            \>\> If $\neg \IsValid(R_{init},l_{init})$ then return $\bot$\\[1pt]
            \>\> If $u_i \in \Cr$ then $\st_{u_i} \gets \st^*_u$\\[1pt]
            \>\> $\res \getsr \protocol{}.\mathsf{Register}(\st_{u_i}, \pp)$\\[1pt]
            \>\> If $\res == \bot$ then Return $\bot$\\[1pt]
            \>\> Else Parse $\res$ as $(\st_{u_i}, id_i, l_i, c_i, \cnt_i)$\\[1pt]
            \>\> $\LLL[i] \gets (id_i, l_i, c_i, \cnt_i)$\\[1pt]
            \>\> If $u_i \in \Cr$ then Return $\st_{u_i}$\\[1pt]
            \>\>\> Else return $u_i$\\ \\ 
            
            \>\underline{$\oAttest(i, \st^*_u)$} \\[3pt]
            \>\> If $u_i \notin \mathcal{U}$ then return $\bot$\\[1pt]
            \>\> If $u_i \in \Cr$ then $\st_{u_i} \gets \st^*_u$\\[1pt] 
            \>\> $\res \getsr \protocol{}.\mathsf{Attest}(\st_{u_i},\sk_{KC})$\\[1pt]
            \>\> If $\res == \bot$ then return $\bot$\\[1pt]
            \>\> Parse $\res$ as $(c'_i, \sigma)$\\[1pt]
            \>\> $\st_{u_i} \gets (\cdot, \cdot, \cdot, c'_i, \sigma )$\\[1pt]
            \>\> If $u_i \in \Cr$ then return $\st_{u_i}$\\[1pt]
            \>\>\> Else return $u_i$ \\ \\

            \>\underline{$\oCorrupt(i)$} \\[3pt]
            \>\> If $u_i \notin \mathcal{U} \vee u_i \in \Cr$ then return $\bot$\\[1pt]
            \>\> $\Cr \gets \Cr \cup \{u_i\}$\\[1pt]
            \>\> Return $\st_{u_i}$
          \end{tabbing}
        \end{small}
      \end{minipage} 
      \begin{minipage}{.24\textwidth}
        \begin{small}\begin{tabbing}
            1\=1\=12\=12\=12345\=\kill
            \underline{$\oUpdate(i,j,S_{real})$} \\[3pt]
            \> If $u_i \notin \mathcal{U} \vee u_j \notin \mathcal{U}$ then return $\bot$\\[1pt]
            \> If $\neg \IsValid(S_{real})$ then return $\bot$\\[1pt]
            \> Parse $\LLL[i]$ as $(id_i, l_i, c_i, \cnt_i)$ \\[1pt]
            \> Parse $\LLL[j]$ as $(id_j, l_j, c_j, \cnt_j)$ \\[1pt]            
            \> $c'_i \getsr \protocol{}.\mathsf{Update}(\st_{SP}, \pp, c_i, c_j)$\\[1pt]
            \> $c'_j \getsr \protocol{}.\mathsf{Update}(\st_{SP}, \pp, c_j, c_i)$\\[1pt]
            \>\> $\msg_i \gets i, \msg_j \gets j$\\[1pt]
            \> For $x \in \{i,j\}$:\\[1pt]
            \>\> $\cnt_x \gets \cnt_x+1$\\[1pt]
            \>\> If $\cnt_x == N$ then \\[1pt]
            \>\>\> $R'_x \getsr \protocol{}.\mathsf{Announce}(\st_{KC},c_i)$\\[1pt]
            \>\>\> 
            $M[i]  \gets R'_x$ \\[1pt]
            \>\>\> If $u_x \in \Cr$ then $\msg_x \gets R'_x$\\[1pt]
            \> Return $(\msg_i,\msg_j)$ \\ \\

            1\=1\=12\=12\=12345\=\kill
            \underline{$\oVerifyNew(i, \st^*_u)$} \\[3pt]
            \> If $u_i \notin \mathcal{U}$ then return $\bot$\\[1pt]
            \> If $u_i \in \Cr$ then $\st_{u_i} \gets \st^*_u$\\[1pt]
            \> $\res \getsr \protocol{}.\mathsf{VerifyNew}(\st_{SP},\st_{u_i},\pp)$\\[1pt]
            \> If $\res == \bot$ then return $\bot$ \\[1pt]
            \> Parse $\res$ as $(c'_i,l'_i)$\\[1pt]
            \> $\cnt_i \gets 1$; $\LLL[i] \gets (\cdot, l'_i, c'_i, \cnt_i)$\\[1pt]
            \> If $u\in \Cr$ then return $(c'_i,l'_i)$ \\[1pt]
            \> $\st_{u_i} \gets (\cdot, l'_i, \cdot, c'_i, \cdot)$\\[1pt]
            \> Return $u_i$
          \end{tabbing}
        \end{small}
      \end{minipage}
      }
    }\caption{The oracles $\mathcal{O}_f$ that $\A$ has access to in $\mathbf{Exp}^{\text{N-fairness}}_{\Pi, \A}$.}\label{fig:foracles}
  \end{figure}
\end{definition}

We now present the security notion of \textbf{Hidden Elo}, which ensures that the server cannot learn a user's exact Elo value beyond the rank to which that user belongs. This security property is defined by an indistinguishability-based game between the adversary $\A$ and a challenge that runs the protocol. 

As in the game that defined N-round fair Matchmaking, the adversary here is allowed to interact with the same suite of oracles modelling registration, updates, attestation, rank update, and user corruption (shown in \autoref{fig:horacles}). However, here the adversary will be provided those information available to \CSP{}. This modification reflects that \CSP{} is semi-honest, which means it faithfully follow the protocol specification but it records and inspects any message it receives and may collude with corrupt users. Additionally, the adversary is given access to a challenge oracle $\oChallenge$. Upon receiving a query specifying a target user and a pair of initial Elo values $(R_0, R_1)$ of the same rank, the challenge selects one according to a random bit $b$ chosen at the beginning of the game, and runs the registration algorithm. After that, it continues the protocol thereafter. Notice that the challenge oracle can be called on a user once, since any user in the system cannot register twice. 

At the end of the interaction, $\A$ outputs a guess $b'$ for the random bit. To prevent trivial wins, $\A$ is not allowed to take over any of the challenged users, since such a takeover would reveal the underlying Elo value with the help of \KC{} and thereby invalidate the intended security guarantee. The adversary wins the game if it guesses the bit correctly. Hidden Elo holds if any PPT adversary's advantage in this game is negligible in the security parameter. 

\begin{definition}[Hidden Elo]\label{def:ind}
A private and secure rating system-based matchmaking protocol $\Pi$ is said to preserve Hidden Elo for a valid predefined set of ranks $\LL$ with a set of user $\mathcal{U}$, for any probabilistic polynomial-time adversary $\A$, the following holds:
    \[
        \Adv^{\text{ind}}_{\Pi, \A}(\lambda) = \big| \Pr[\Exp^{\text{ind}}_{\Pi, \A}(\lambda) \rightarrow 1] - \frac{1}{2} \big|
    \]
  is a negligible function in $\lambda$, and the experiment $\Exp^{\text{ind}}_{\Pi, \A}$ is defined as follows: 

\smallskip  

  \begin{minipage}{\textwidth}
    \begin{small} \begin{tabbing}
        12\=12\=123\=12\=\kill
        \underline{$\mathbf{Exp}^{\text{ind}}_{\Pi, \A}(\lambda)$} \\[3pt]
        \> $b \getsr \{0,1\}$; $\LLL, \Cr, \chall \gets \bot$\\[1pt]
        \> $(\st_{SP}, \st_{KC}, \{\st_{u_i}\}_{u \in \mathcal{U}}, \pp) \getsr \protocol{}.\mathsf{Init}(1^\lambda, \mathcal{R}_{rccc})$\\[1pt]
        \> $b' \getsr \A(1^\lambda, \pp: \mathcal{O}_{h})$ \\[1pt]
        \> Return $(b' == b)$
      \end{tabbing}\end{small}
  \end{minipage}

  \begin{figure}[h]\scalebox{0.9}{
    \framebox{
      \hfill \begin{minipage}{.24\textwidth}
        \begin{small}\begin{tabbing}
            1\=1\=12\=12\=12\=\kill
            \>\underline{$\oRegister(i,\st^*_{u_i},R_{init},l_{init})$} \\[3pt]
            \>\> If $\LLL[i] \neq \emptyset \vee u_i \notin \mathcal{U}$ then return $\bot$\\[1pt]
            \>\> If $\neg \IsValid(R_{init},l_{init})$ then return $\bot$\\[1pt]
            \>\> If $u_i \in \Cr$ then $\st_{u_i} \gets \st^*_u$\\[1pt]
            \>\> $\res \getsr \protocol{}.\mathsf{Register}(\st_{u_i}, \pp)$\\[1pt]
            \>\> If $\res == \bot$ then Return $\bot$\\[1pt]
            \>\> Else Parse $\res$ as $(\st_{u_i}, id_i, l_i, c_i, \cnt_i)$\\[1pt]
            \>\> $\LLL[i] \gets (id_i, l_i, c_i, \cnt_i)$\\[1pt]
            \>\> If $u_i \in \Cr$ then Return $\st_{u_i}$\\[1pt]
            \>\>\> Else return $(l_i, c_i, \cnt_i)$\\ \\ 
            
            \>\underline{$\oAttest(i, \st^*_u)$} \\[3pt]
            \>\> If $u_i \notin \mathcal{U}$ then return $\bot$\\[1pt]
            \>\> If $u_i \in \Cr$ then $\st_{u_i} \gets \st^*_u$\\[1pt] 
            \>\> $\res \getsr \protocol{}.\mathsf{Attest}(\st_{u_i},\sk_{KC})$\\[1pt]
            \>\> If $\res == \bot$ then return $\bot$\\[1pt]
            \>\> Parse $\res$ as $(c'_i, \sigma)$\\[1pt]
            \>\> $\st_{u_i} \gets (\cdot, \cdot, \cdot, c'_i, \sigma )$\\[1pt]
            \>\> If $u_i \in \Cr$ then return $\st_{u_i}$\\[1pt]
            \>\>\> Else return $u_i$ \\ \\

            \>\underline{$\oChallenge(i, R_0, R_1, l)$} \\[3pt]
            \>\> If $u_i \notin \mathcal{U} \vee u_i \in \Cr$ then return $\bot$\\[1pt]
            \>\> If $u_i \in \chall$ then return $\bot$\\[1pt]
            \>\> If $\neg \IsValid(R_0, l) \vee \neg \IsValid(R_1, l)$ \\[1pt] 
            \>\>\> Then return $\bot$ \\[1pt]
            \>\> If $\textsf{rank}(R_0) \neq \textsf{rank}(R_1)$ then return $\bot$ \\[1pt]
            \>\> $\res \getsr \protocol.\mathsf{Register}(\st_{u_i}, R_b, l, \pp)$\\[1pt]
            \>\> If $\res == \bot$ then return $\bot$ \\[1pt]
            \>\> Parse $\res$ as $(\st_{u_i}, id_i, l_i, c_i, \cnt_i)$\\[1pt]
            \>\> $\chall \gets \chall \cup \{u_i\}$\\[1pt]
            \>\> Return $(l_i, c_i, \cnt_i)$
          \end{tabbing}
        \end{small}
      \end{minipage} 
      \begin{minipage}{.24\textwidth}
        \begin{small}\begin{tabbing}
            1\=1\=12\=12\=12345\=\kill
            \underline{$\oUpdate(i,j,S_{real})$} \\[3pt]
            \> If $u_i \notin \mathcal{U} \vee u_j \notin \mathcal{U}$ then return $\bot$\\[1pt]
            \> If $\neg \IsValid(S_{real})$ then return $\bot$\\[1pt]
            \> Parse $\LLL[i]$ as $(id_i, l_i, c_i, \cnt_i)$ \\[1pt]
            \> Parse $\LLL[j]$ as $(id_j, l_j, c_j, \cnt_j)$ \\[1pt]            
            \> $c'_i \getsr \protocol{}.\mathsf{Update}(\st_{SP}, c_i, c_j, \pp)$\\[1pt]
            \> $c'_j \getsr \protocol{}.\mathsf{Update}(\st_{SP}, c_j, c_i, \pp)$\\[1pt]
            \> For $x \in \{i,j\}$:\\[1pt]
            \>\> $\cnt_x \gets \cnt_x+1$\\[1pt]
            \>\> If $\cnt_x == N$ then \\[1pt]
            \>\>\> $R'_x \getsr \protocol{}.\mathsf{Announce}(\st_{KC},c_i)$\\[1pt]
            \> Return $(c'_i, c'_j, \cnt_i, \cnt_j)$ \\ \\

            1\=1\=12\=12\=12345\=\kill
            \underline{$\oVerifyNew(i, \st^*_u)$} \\[3pt]
            \> If $u_i \notin \mathcal{U}$ then return $\bot$\\[1pt]
            \> If $u_i \in \Cr$ then $\st_{u_i} \gets \st^*_u$\\[1pt]
            \> $\res \getsr \protocol{}.\mathsf{VerifyNew}(\st_{SP},\st_{u_i},\pp)$\\[1pt]
            \> If $\res == \bot$ then return $\bot$ \\[1pt]
            \> Parse $\res$ as $(c'_i,l'_i)$\\[1pt]
            \> $\cnt_i \gets 1$; $\LLL[i] \gets (\cdot, l'_i, c'_i, \cnt_i)$\\[1pt]
            \> If $u\in \Cr$ then return $(c'_i,l'_i)$ \\[1pt]
            \> $\st_{u_i} \gets (\cdot, l'_i, \cdot, c'_i, \cdot)$\\[1pt]
            \> Return $(c'_i, l'_i)$ \\ \\
            
            \>\underline{$\oCorrupt(i)$} \\[3pt]
            \> If $u_i \notin \mathcal{U} \vee u_i \in \Cr$ then return $\bot$\\[1pt]
            \> If $u_i \in \chall$ then return $\bot$ \\[1pt]
            \> $\Cr \gets \Cr \cup \{u_i\}$\\[1pt]
            \> Return $\st_{u_i}$
          \end{tabbing}
        \end{small}
      \end{minipage}
      }
    }\caption{The oracles $\mathcal{O}_h$ that $\A$ has access to in $\mathbf{Exp}^{\text{ind}}_{\Pi, \A}$.}\label{fig:horacles}
  \end{figure}
\end{definition}

\begin{theorem}\label{thm:fairness}
If the signature scheme $\Sigma$ is EUF-CMA secure and the range proof $\NIZK$ is knowledge-sound, the construction \protocol{} ensures N-round Fair Matchmaking. 
\end{theorem}

\begin{proof}
To win the N-round Fair Matchmaking game, the adversary $\A$ must produce a valid range proof together with a valid attestation that jointly authenticate a user's identifier, an encrypted Elo value and the corresponding commitment bound to that value. There are only two possible ways in which the adversary can succeed in doing so:
\begin{enumerate}
    \item the adversary provides a valid attestation on an invalid Elo value, namely, a value that does not match the user's current Elo value as that maintained by \KC{}, or
    \item the adversary presents a valid attestation, but provide it with a range proof that verifies successfully for an invalid Elo value, even though this value is consistent with the commitment and ciphertext specified in the attestation.
\end{enumerate}

We discuss these two cases separately. In the first case, recall that \KC{} is assumed to be semi-honest and therefore maintains an accurate record of each user's current Elo value. By design, it only attests to commitments and ciphertexts that correspond exactly to those current values. Hence, it will never sign or attest to an inconsistent tuple $(id_i, c_i, C_i)$ unless an adversary can forge a valid signature under \KC{}'s signing key. Consequently, any adversary that succeeds in this case can be used to construct a forger that breaks the EU-CMA security of the signature scheme.

For the second case, $\A$ must produce a range proof that verifies under the relation $\mathcal{R}_{rccc}$ for a ciphertext and commitment pair that does not encode a valid Elo value. Because the attestation in this scenario is valid, the only remaining violation would become the ability to construct a convincing proof for a false statement in the underlying NIZK system. Therefore, any adversary that succeeds in this case can be transformed into an adversary that breaks the knowledge soundness property of the NIZK proof system $\NIZK$. 

Bringing them together, these two reductions bound the adversary's overall advantage in the N-round Fair Matchmaking game. Formally, there exists efficient adversaries $\B$ and $\C$ such that
\[
\Adv^{\text{N-fairness}}_{\protocol{},\A}(\lambda) \leq \Adv^{\text{EU-CMA}}_{\Sigma,\B}(\lambda) + \Adv^{\text{soundness}}_{\NIZK,\C}(\lambda),
\]
where $\Adv^{\text{EU-CMA}}_{\Sigma,\B}(\lambda)$ and $\Adv^{\text{soundness}}_{\NIZK,\C}(\lambda)$ denote the advantages in the signature scheme's EU-CMA security game and the range proof's knowledge soundness game respectively. Thus, given that both primitives are secure, the construction \protocol{} preserves the N-round Fair Matchmaking property. 
\end{proof} 

\begin{theorem}\label{thm:ind}
The commitment scheme $\Com$ is computationally hiding, the range proof $\NIZK$ for the relation $\R_{rccc}$ is zero-knowledge, and the homomorphic encryption scheme $\HE$ is IND-CPA secure, the construction \protocol{} ensures Hidden Elo. 
\end{theorem}
\begin{proof}
Since the adversary receives precisely \CSP{}'s view and also the corrupted users' during the honest execution of the protocol, we can narrow the leakage points of the protocol to see what can reveal information regarding the exact Elo values. Recall that the \CSP{} observes commitments, ciphertexts, and range proofs for the relation $\R_{rccc}$, while the information obtained from the corrupt users will not help the adversary to learn any information about the exact Elo values of the challenged users. It should be notice that since the commitment is a part of the public statement of the range proof, the zero-knowledge property of the range proof itself does not guarantee that no information about the associated Elo value is leaked through the commitment. Thus, information about an honest user's precise Elo value can leak only through: (1) commitments, (2) ciphertexts, and (3) range proofs. We proceed the proof in the following sequence of games:

\noindent $\game_0$: The initial game is the original security game that defines Hidden Elo. 

\noindent $\game_1$: In this game, we replace the real proofs by simulated proofs for the same public statements. By the zero-knowledge property, the adversary's view in $\game_0$ is indistinguishable from $\game_1$. 

\noindent $\game_2$: This game is identical to the previous one, except that we replace each ciphertext generated in the registration algorithm to the encryption of a random Elo value $R^*$ of the same rank. Since all proofs in $\game_1$ are simulated and thus do not depend on any particular witness, replacing the original ciphertexts will not cause any problem. The IND-CPA security guarantees that this replacement is indistinguishable to the adversary. Notice that we only replace the ciphertext with those honestly generated by running the encryption algorithm on random Elo value of the same rank, so the public statement remains valid in the relation. 

\noindent $\game_3$: Here, we replace each corresponding commitment by one computed from $R^*$ and a fresh randomness $r^*$. Because the proofs are simulated and and the commitment scheme is hiding, the replacement is valid and indistinguishable.


Now the adversary's view depends on $R^*$ and it is independent of the random bit $b$. In such case, it can win the game with probability no better than a half. Thus, given that the commitment scheme is hiding, the range proof is zero-knowledge and the homomorphic encryption scheme is IND-CPA secure, our construction \protocol{} preserves Hidden Elo. 
\end{proof}

\section{Evaluation}
\label{sec:eval-results}
We have not benchmarked the range proof separately. In our protocol, the service provider performs many homomorphic evaluations per game, while the user generates at most one proof per $N$ games. Accordingly, the amortised online cost is expected to be dominated by the FHE computation, although the precise trade-off depends on the instantiation. The NIZK considered here is a composed one, combining a commitment-based range proof with a proof that the FHE ciphertext encrypts the same value. It is most naturally instantiated via a lattice-native NIZK~\cite{LyubashevskyNP22}. 
We leave the instantiation and benchmarking of the range proof as future work.


We evaluate the usability of \protocol{} by measuring the runtime, memory usage and accuracy of the FHE-based update mechanism proposed in~\Cref{alg:update} in a chess-based scenario. 
We implement the algorithm with the OpenFHE 
library~\cite{badawi22openfhe} and benchmark it on a standard desktop machine with a~12th Generation Intel Core i7-12700@3.7GHz CPU,~32GB RAM, running Ubuntu~24.04.2 LTS. 

Runtime and memory usage experiments were performed~500 times (single-thread) and the average results for each parameter set is provided in~\autoref{tbl:eval-time-toy} and \autoref{tbl:eval-memory-toy}, respectively. Accuracy results, provided in \autoref{tbl:eval-acc-toy}, are calculated after~10 000 consecutive rating updates on a randomly initiated rating. More specifically, these results show the average difference between the new rating approximated with CKKS-RNS using \Cref{alg:update} (we provide details of the polynomial approximation in~\autoref{app:poly-appro}) and the same update computed on plaintext through \autoref{eq:elo-rating}. The opponents' ratings ($n=3$) are sampled randomly from the same distribution as the player and the match outcomes are picked randomly between loss (0), draw (0.5) and win (1). Additionally, we cover the results of multi-threading in \autoref{app:multi-thread}. We run our experiments on three parameter sets with different security levels: 
\begin{enumerate}
    \item security level $\lambda\approx12$ -- a toy parameter set for showcasing the functionality of the implementation.
    \item security level $\lambda\approx80$ -- custom parameter set estimated to $\lambda = 80$ bits~\cite{albrecht2015concrete}.
    \item security level $\lambda\approx128$ -- standard 128-bit security level parameter set~\cite{albrecht2021homomorphic}.
\end{enumerate}


\noindent \textbf{Open Science \& Reproducible Research.} \enskip 
To support open and reproducible research, we provide our source code on GitHub\footnote{\href{https://github.com/UnoriginalOrigi/mm_elo_HE}{\protocol{} GitHub Repository}}.

\subsection{Runtime Analysis}
\label{subsec:runtime-results}

\begin{table}[ht!]
\caption{Average runtime for each operation in milliseconds (ms) over 500 iterations ($n = 3$).}
\small
\centering
\begin{tblr}{
  width = \linewidth,
  colspec = {Q[350]Q[170]Q[200]Q[180]},
  row{1} = {azure},
  cell{1}{2} = {c},
  cell{1}{3} = {c},
  cell{1}{4} = {c},
  cell{2-13}{2-4} = {r},
  hlines,
  vlines,
  hline{1-2,5,11-12,14} = {0.2em},
  vline{1-2,5} = {0.2em},
}
\diagbox[innerwidth=\linewidth]{\color{white}\textbf{Operation}}{\color{white}\textbf{Sec Level}} & \color{white}$\lambda=12$ & \color{white}$\lambda=80$ & $\color{white}\lambda=128$ \\

Generate $\pk$ and $\sk$    & 12.413   & 204.368      & 418.347    \\
Generate $\mathsf{rlk}$       & 32.656   & 541.294      & 1120.767   \\
Generate $\mathsf{bsk}$       & 78.828   & 1534.842     & 3332.135   \\

Encryption              & 2.396    & 45.868       & 99.446     \\
Decryption              & 2.195    & 27.444       & 55.088     \\
EvalMul (Constant)      & 0.087    & 1.722        & 4.438      \\    
EvalAdd                 & 0.180    & 3.840        & 8.457      \\
EvalSub                 & 0.068    & 0.803        & 2.430      \\
EvalPoly                & 59.853   & 1035.457     & 2171.868   \\

Bootstrapping           & 328.714  & 7449.077     & 15785.967  \\

Rating Update (RU)      & 497.483    & 11465.392    & 23247.823 \\
RU w\textbackslash o Bootstrapping     & 168.769    & 4016.315     & 7461.856 \\
\end{tblr}
\label{tbl:eval-time-toy}
\end{table}

\autoref{tbl:eval-time-toy} showcases the main runtime benchmarks: \begin{enumerate*} [\itshape (i)]
    \item key generation costs ($\pk/\sk$, $\rlk$, $\bsk$),
    \item individual operation costs, and
    \item one full update cost. 
\end{enumerate*}
The key generation costs consist of three parts:
\begin{enumerate}
    \item Generate $\pk$ and $\sk$ -- generates the public/private key pair ($\mathsf{pk/sk}$) for homomorphic encryption and decryption.
    \item Generate $\mathsf{rlk}$ -- generates the relinearization key ($\mathsf{rlk}$) for reducing the size of the ciphertext after multiplication and other operations.
    \item Generate $\mathsf{bsk}$ -- generates the bootstrapping key ($\mathsf{bsk}$) for refreshing the noise of a homomorphic ciphertext.
\end{enumerate}

Moreover, the individual costs consist of algorithms used for computing a rating update: \begin{enumerate*}[(i)]
    \item encryption,
    \item decryption,
    \item bootstrapping,
    \item polynomial approximation (details provided in \autoref{app:poly-appro}),
    \item multiplication,
    \item addition, and
    \item subtraction.
\end{enumerate*}

Following this, it can be observed that the largest computation costs for \protocol{} are seen in the key generation (particularly "Generate $\mathsf{bsk}$") as well as the Bootstrapping algorithms. The combined key generation algorithms are the second most computationally expensive algorithms in the scheme, just behind the Bootstrapping algorithm for all parameter sets. More specifically, the combined time costs for key generation are~0.12s,~2.28s and~4.87s, respectively for each parameter set. However, key generation needs to be run only once during registration for the functionality of \protocol{}. As such, these costs are reasonable in most scenarios as they will not introduce consistent latency for users of the application. 

On the other side, the bootstrapping algorithm is the most expensive part in \protocol{} taking on average~15.79s to compute on security level $\lambda=128$ and a full update takes~23.25s on average. While the computation time is considerably longer than a plaintext alternative (which takes a fraction of a second), the benefits come in the additional security that is provided and keeps the rating secret from malicious users who aim to maximize their odds of obtaining a beneficial result. More importantly, the individual costs of computing a rating update without bootstrapping are low as even on the highest security level parameter set it takes approximately~7.46s to compute an update for a user. In practice, this computational cost would be divided into smaller parts, where \CSP{} could blindly compute part of the update after each match up to $n-1$. In this case, $u_i$ would only have to wait for the final update to be computed after his latest match. In our benchmarking scenario, that would have a latency of approximately~17.99s for $u_i$ on security level $\lambda=128$ (a single run of~\Cref{alg:update}, where $j==n-1$), despite the $n$ value chosen as all previous updates up to $n-1$ could be computed after each match.

\smallskip

\noindent\textbf{Considerations.} \enskip From these experiments, we can conclude that \protocol{} does not allow for constructing usecases which would require instant access to the updated ratings. However, the costs are reasonable for implementing it in less time-sensitive applications, such as for online marketplace reputation systems and online tournaments, where ratings are updated after the tournament concludes.

\subsection{Memory Analysis}
\label{subsec:memory-results}

\begin{table}[ht!]
\caption{Average memory costs for each operation in megabytes (MB) over 500 iterations ($n = 3$).}
\small
\centering
\begin{tblr}{
  width = \linewidth,
  row{1} = {azure},
  colspec = {Q[450]Q[150]Q[150]Q[150]},
  cell{1}{2} = {c},
  cell{1}{3} = {c},
  cell{1}{4} = {c},
  cell{2-15}{2-4} = {r},
  hlines,
  vlines,
  hline{1-2,4,7-8,14-15} = {0.2em},
  vline{1-2,5} = {0.2em},
}
\diagbox[innerwidth=\linewidth]{\color{white}\textbf{Operation}}{\color{white}\textbf{Sec Level}} & \color{white}$\lambda=12$ & \color{white}$\lambda=80$ & $\color{white}\lambda=128$ \\

Generate Context     & $5.70$  & $82.80$    & $165.55$  \\
Precompute Bootstrapping    & $2.14$  & $28.06$    & $55.42$  \\

Generate $\pk$ and $\sk$     & $6.88$  & $129.12$    & $259.04$  \\
Generate $\mathsf{rlk}$    & $12.68$  & $197.79$    & $395.36$  \\
Generate $\mathsf{bsk}$      & $163.04$  & $3538.72$   & $7548.96$ \\

One Ciphertext         & $1.04$   & $15.70$   & $29.56$ \\

Encryption                & 0.00        & 0.00      & 0.00    \\
Decryption                & 0.00        & 0.00      & 0.00       \\
EvalMul (Constant)        & 0.32        & 0.00      & 0.00        \\    
EvalAdd                   & 0.00        & 4.80      & 2.24       \\
EvalSub                   & 0.00        & 0.00      & 0.00        \\
EvalPoly                  & 3.80        & 33.90     & 81.21       \\

Bootstrapping             & 32.17       & 21.14     & 20.10    \\  
\end{tblr}
\label{tbl:eval-memory-toy}
\end{table}

Similarly to~\autoref{subsec:runtime-results}, the memory cost benchmarks presented in \autoref{tbl:eval-memory-toy} are divided into three parts: \begin{enumerate*} [\itshape (i)]
    \item context generation and precomputation memory usage,
    \item key generation memory usage, and
    \item individual operation memory usage.
\end{enumerate*}

From these results, the core memory overhead of \protocol{} is from setting up the parameter sets as well as generating the required keys. More specifically, the highest costs comes from generating and storing the bootstrapping key ($\bsk$) with the costs reaching approximately 7 549 MB (7.5 GB) of memory on the highest security level. While memory costs for generating the cryptographic contexts and keys are high, they are manageable by \CSP{}. The keys can be serialized and stored securely on a database and only loaded into memory when a rating update is required for a user. With current capabilities of servers with terabytes of storage the approach is scalable to moderate user base sizes. Future improvements to HE key sizes, through bootstrapping key optimizations~\cite{agrawal2024heap}, would improve the scalability of the approach even further.

Conversely, the memory costs of the individual operations are negligible, varying between~0.00 MB and~0.081 MB throughout the individual operations
for security level $\lambda=128$. In lower security levels the variation was even lower. 
Due to this, it is possible to conclude that the individual operations \textbf{do not} contribute to the memory usage.

\smallskip

\noindent\textbf{Considerations.} \enskip The overall memory costs are comprised of preparing the parameter set and generating the needed keys for performing all computations. In view of these results, each application case needs to decide how to implement key management. With a smaller user bases, each user can generate their own keys while retaining proper functionality. Conversely, usecases with large user bases would require an alternative solution, such as having trusted, intermediate parties for key management (for example team managers who could reveal information to their players).

\subsection{Accuracy}
\label{subsec:accuracy-results}

\begin{table}[ht!]
\caption{Rating difference plaintext and HE computed updates over 10,000 consecutive updates ($n = 3$). Lower is better.}
\small
\centering
\begin{tblr}{
  width = \linewidth,
  row{1} = {azure},
  colspec = {Q[330]Q[200]Q[200]Q[200]},
  cell{1}{2-4} = {c},
  cell{2-5}{2-4} = {r},
  hlines,
  vlines,
  hline{1-2,6} = {0.2em},
  vline{1-2,5} = {0.2em},
}
\diagbox[innerwidth=\linewidth]{\color{white}\textbf{Metric}}{\color{white}\textbf{Sec Level}} & \color{white}$\lambda=12$ & \color{white}$\lambda=80$ & $\color{white}\lambda=128$ \\
Mean difference         & $0.620 \cdot 10^{-4}$   &  $1.990 \cdot 10^{-4}$  &  $5.569 \cdot 10^{-4}$  \\
Standard deviation      & $0.463 \cdot 10^{-4}$   &  $1.601 \cdot 10^{-4}$  &  $4.631 \cdot 10^{-4}$  \\
Min difference      & $0.0004 \cdot 10^{-4}$   &  $0.000$  &  $0.0004 \cdot 10^{-4}$  \\
Max difference      & $2.715 \cdot 10^{-4}$   &  $12.94 \cdot 10^{-4}$  &  $34.92 \cdot 10^{-4}$  \\
\end{tblr}
\label{tbl:eval-acc-toy}
\end{table}

CKKS is an approximate HE scheme, meaning arithmetic operations introduce noise to the result, which may deteriorate the precision of the underlying plaintext. As \protocol{} makes use of CKKS for rating updates, it is important to measure this variance between plaintext and encrypted computations. Following this, \autoref{tbl:eval-acc-toy} provides four metrics to measure the deviation from plaintext results (the estimated precision is covered in \autoref{app:prec-bits}). Specifically, the table contains: 
\begin{enumerate*} [\itshape (i)]
    \item mean difference,
    \item standard deviation,
    \item minimum difference, and
    \item maximum difference.
\end{enumerate*}



Our approach shows negligible deviation from plaintext results, as the mean difference between the plaintext and encrypted results are only~$5.569\cdot10^{-4}$ with the maximum difference being~$34.92\cdot10^{-4}$ on some rating updates. Additionally, these results were achieved without using more computationally expensive approaches, such as iterative bootstrapping~\cite{lee2021high}. The standard deviation is equivalently low, resulting in a deviation of~$4.631\cdot10^{-4}$ after~10 000 consecutive rating updates. Additionally, the minimum difference reached as low as~0.000 (result showed all zeroes from 8 digits in the fractional part) after 10,000 consecutive updates.

As such, we can conclude that while CKKS does introduce some noise to the computations, the overall level of the noise remains low and \textbf{does not} introduce any meaningful deviation from the exact plaintext results. While the rating is an approximation, a difference of~$5.569\cdot10^{-4}$ on average would not affect the fairness of matchmaking even after~10 000 consecutive updates.

\subsection{Comparison to Other Works}

\begin{center}
    \begin{tcolorbox}[colback=yellow!10!white,colframe=red!75!black,title=\textbf{Lack of Similar Works as an Indicator of Novelty}]
        Comparing \protocol{} with different private matchmaking techniques introduces unique difficulties for fair comparisons. For example, works on PSI~\cite{he2022differentially,cong2021labeled} focus on finding and \textit{revealing} unique items which are the same between two individual sets, leading to a conceptually different approach than \protocol{} as in our work the focus is on keeping a single value (the rating) hidden even after finding a match. As such, there are few works -- most notably Matchmaking Encryption~\cite{ateniese2021match} -- with similar goals which would allow for a meaningful comparison with \protocol{}. 
    \end{tcolorbox}
\end{center}

One work that we identified that can be used as a point of comparison is ME~\cite{ateniese2021match}. ME aims to match two users who have individual attributes and can set their own policies for a ciphertext to be revealed only when an exact match occurs between the attributes and policies. While different from \protocol{}, ME could be used to implement a matchmaking system between two similarly matched opponents following the same flow as \protocol{}. Some of the main differences between ME and \protocol{} is the underlying primitives used, the trust assumptions and the performance. 

Firstly, \protocol{} provides post-quantum security as FHE is constructed with the quantum-resistant RLWE assumption, conversely the proposed construction of ME relies on the Bilinear Diffie-Hellman (BDH) assumption, which is vulnerable to quantum attacks. 

Secondly, 
ME assumes the existence of a Trusted Authority (TA) that 
can have access to each players policies and attributes and is responsible for generating the keys associated to them. 
\protocol{} also makes use of a TA, however, the only function it performs is attesting that the ciphertext linked to a users new rating is generated correctly. Consequently, the trust assumptions required in ME are much stronger than those required by \protocol{}. 

Lastly, looking at the performance of the two works, ME~\cite{ateniese2021match} outperforms \protocol{} by a considerable margin even without taking into account the rating update computations as shown in \autoref{tbl:ME-comp}. However, 
this gap is something to be expected as \protocol{} makes use of 
RLWE, and therefore provides post-quantum security, while ME is implemented using BDH which makes the underlying construction more efficient but offers weaker security guarantees. 
Additionally, ME can run into the same issues as plaintext solutions since a players rating can be kept hidden from them as only the TA needs to know the rating for attribute certification, resulting in a player being unsure of his exact rating, which \protocol{} avoids. 

\begin{table}[h]
\caption{Comparison between \protocol{} and Matchmaking Encryption. Computation time in miliseconds (ms)}
\centering
\begin{tblr}{
  width = \linewidth,
  colspec = {Q[429]Q[404]Q[404]},
  row{1} = {azure},
  column{2-3} = {r},
  cell{1}{2-3} = {c},
  hlines,
  vlines,
}
           & \color{white}\textbf{ME} ($\lambda = 80$) & \color{white}\protocol{} ($\lambda = 80$) \\
Key Generation      & 7.867 & 16110.504    \\
Encryption          & 7.012 & 45.868    \\
Decryption          & 4.385 & 27.444   
\end{tblr}
\label{tbl:ME-comp}
\end{table}

To conclude, while \protocol{} has considerably higher computational costs than ME, \protocol{} provides 
stronger security guarantees and lowers the overall trust assumptions needed for implementing rating-based matchmaking. Naturally, a service provider would have to weigh their own capabilities and needs when choosing how to implement a private matchmaking solution and what privacy leakages are acceptable in their individual scenario.

\section{Conclusion}
\label{sec:conclusions}
In this paper, we introduced \protocol{}, a rating system-based matchmaking protocol by making use of HE and NIZK proving system. To the best of our knowledge, this is the first work to tackle private and secure matchmaking through a rating system without obfuscating a users ratings from the user themselves. From our analysis we show that there is latency can conclude that our approach can be applied to real world scenarios without any noticeable latency for users. Additionally, we prove \protocol{} allows for fair matchmaking against a possible adversary.



\textbf{Future Directions.} 
This paper provides the foundation for exploring private matchmaking through rating system. We believe that our work can further be expanded into other rating systems, such as Glicko and we leave this research as a future work. Alternative constructions could focus on performing the rating updates in a decentralized setting, namely instead of outsourcing the computation to a CSP, the user could make the update locally and prove the correctness of the update process and output through zk-SNARKS. Further work could make use of scheme switching in HE to allow for simpler encoding for building an efficient range proof system. 


\begin{acks}
    This work was funded by the EU research project SWARMCHESTRATE (No. 101135012). We thank Dr. Zichen Gui for the discussions and his advice during the writing process. We extend our gratitude to all the reviewers and the valuable feedback they provided to improve this work.
\end{acks}


\printbibliography

\appendix
\begin{figure*}[ht!]
    \begin{subfigure}[b]{.40\textwidth}
\begin{tikzpicture}
		\begin{axis}[
    			width=\linewidth,
    			height=.6\linewidth,
    			xlabel={Users},
    			ylabel={Runtime (s)},
    			xmin=0, xmax=310,
    			ymin=0, ymax=820,
    			xtick={0, 50, 100, 150, 200, 250, 300},
    			ytick={0, 200, 400, 600, 800},
    			legend pos =north west,
    			ymajorgrids=true,
    			grid style = dashed,
            ]
			\addplot[
    			color=blue,
    			mark=square,
    			]
    			coordinates {
    				(1,0.384)(50,3.299)(100,6.292)(150,9.208)(200,12.189)(250,15.145)(300,18.159)
    			};
			\addplot[
    			color=red,
    			mark=square,
    			]
    			coordinates {
    				(1,8.780)(50,68.287)(100,128.621)(150,189.420)(200,249.282)(250,307.585)(300,374.827)
    			};
                \addplot[
    			color=forestgreen,
    			mark=square,
    			]
    			coordinates {
    				(1,18.515)(50,144.161)(100,270.701)(150,397.101)(200,525.792)(250,664.027)(300,798.998)
    			};
			\legend{$\lambda = 12$, $\lambda = 80$, $\lambda = 128$}
		\end{axis}
	\end{tikzpicture}
	\caption{Runtime per Opponent Amount}
	\label{sfig:time-user}
    \end{subfigure}
    \begin{subfigure}[b]{.40\textwidth}
    	\begin{tikzpicture}
		\begin{axis}[
    			width=\linewidth,
    			height=0.6\linewidth,
    			xlabel={Users},
    			ylabel={Memory Usage (MB)},
    			xmin=0, xmax=310,
    			ymin=0, ymax=8500,
    			xtick={0, 50, 100, 150, 200, 250, 300},
    			ytick={0, 2000, 4000, 6000, 8000},
    			ymajorgrids=true,
    			grid style = dashed,
            ]
			\addplot[
    			color=blue,
    			mark=square,
    			]
    			coordinates {
    				(1,20.156)(50,66.396)(100,109.788)(150,163.284)(200,214.064)(250,262.484)(300,335.036)
    			};
			\addplot[
    			color=red,
    			mark=square,
    			]
    			coordinates {
    				(1,4.136)(50,565.540)(100,1212.900)(150,1854.568)(200,2491.208)(250,3246.724)(300,4159.440)
    			};
                \addplot[
    			color=forestgreen,
    			mark=square,
    			]
    			coordinates {
    				(1,2.968)(50,1133.776)(100,2438.932)(150,3692.300)(200,4977.900)(250,6479.836)(300,8312.484)
    			};
		\end{axis}
	\end{tikzpicture}
	\caption{Memory Usage per Opponent Amount}
	\label{sfig:memory-user}
    \end{subfigure}
    \caption{Opponent Number Cost Impact}
    \label{fig:opp-impact}
\end{figure*}

\section{Design Rationale}
\label{ssec:design-rat}

To implement H-Elo it was important to decide on the HE scheme we would utilise for the homomorphic operations. The main schemes for consideration were \begin{enumerate*}[\itshape(i)]
    \item B/FV;
    \item BGV;
    \item CKKS;
    \item TFHE.
\end{enumerate*}
Additionally, the functionalities we required are as follows: \begin{enumerate}
    \item Ability to compute an Elo rating update, which consists of \begin{enumerate*}[\itshape(i)]
        \item two additions;
        \item two subtractions;
        \item two multiplications with a constant;
        \item a polynomial approximation of~\autoref{eq:exp-result}.
    \end{enumerate*}
    \item Bootstrapping capabilities;
    \item Allow values from~0 to~4000.
\end{enumerate}

Following our requirements we analysed the main HE schemes. Beginning with TFHE, while it is optimized to perform bootstrapping on sparsely packed ciphertexts, TFHE uses logic gates for homomorphic operations. Due to this, arithmetic operations become more complex to implement compared to other schemes. Additionally, as TFHE performs bootstrapping after each gate, the computational and memory costs of TFHE grow rapidly for each arithmetic operation that is performed~\cite{tsuji2024comparison}. Additionally, the performance of TFHE is the best when the values are bound to~8-bit integers ($2^8$). While values upwards of~$2^{32}$ are possible, the performance of the scheme degrades and takes longer for the computations than other schemes~\cite{al2023demystifying,tsuji2024comparison}. As such, we excluded TFHE as a choice.

The remaining schemes are the exact HE schemes B/FV and BGV, as well as CKKS, an approximate HE scheme. Both exact schemes have similar capabilities for integer-based homomorphic operations and fit our requirements. However, the core reasoning for not choosing these schemes comes from performance, namely integer-based polynomial approximation of~\autoref{eq:exp-result} ends up being both costly and inaccurate, with the polynomial depth being limited and the result potentially becoming more inaccurate than CKKS. Additionally, another reason for choosing CKKS was the cost of bootstrapping. While BGV does have better bootstrapping capabilities in terms of performance for sparsely packed ciphertext (to a point), it shows lower precision of the resulting refreshed ciphertext than CKKS. This could introduce a more pronounced error between the plaintext and encrypted result. Additionally, following comparisons of CKKS and BGV, we can notice that BGV performs worse when bootstrapping is required~\cite{al2023demystifying} than CKKS. Due to these reasons, we choose CKKS over other HE schemes. 

\section{Notation}
\label{app:notation}


We denote the HE algorithms for key generation, encryption and decryption as we defined them in~\autoref{def:he}. 
We expand the evaluation algorithm into separate algorithms respective to their underlying function. More specifically: \begin{itemize}
    \item $\mathsf{EvalAdd}(c_x, c_y)$ -- is the addition algorithm, which evaluates $c_x + c_y$;
    \item $\mathsf{EvalSub}(c_x, c_y)$ -- is the subtraction algorithm, which evaluates $c_x - c_y$;
    \item $\mathsf{EvalMul}(c_x, c_y)$ -- is the multiplication algorithm, which evaluates $c_x \cdot c_y$;
    \item $\mathsf{EvalPoly}(c_x, f)$ -- is the polynomial approximation algorithm, which uses Chebyshev approximations to evaluate a specific function $f(c_x)$.
    \item $\mathsf{EvalBootstrap(}c)$ -- is the bootstrapping algorithm, which refreshes the accumulated noise of a ciphertext $c$ to allow further computations. More specifically, it can be defined as homomorphically evaluating the decryption circuit with an encrypted secret key (also known as the bootstrapping key $\mathsf{bsk = HE.Enc(pk, sk)}$), namely $\mathsf{HE.Bootstrap(}c) = \mathsf{HE.Eval(evk,HE.Dec(\cdot)},$ $\mathsf{bsk}, c)$.
\end{itemize}  

\section{Multi-threaded Results}
\label{app:multi-thread}
Increasing the number of matches required for a rating update increases the runtime costs linearly as can be seen in \autoref{sfig:time-user}. The protocol execution can also be \textit{multi-threaded} resulting in lower runtime for a rating update. In comparison, a single-threaded execution takes approximately~799s, while a \textit{multi-threaded} execution takes~297s (approximately~$60\%$ decrease in runtime), when $n = 300$ and $\lambda = 128$. Naturally, these results would continue to hold true as $n$ increases, but the efficiency of multi-threading is reduced when $n$ is reduced. It is important to note, that these results show the total time to fully compute the new rating after all matches concluded, however as noted previously, this would not be required in most applications, as the expected outcomes ($c_{exp}$ and their sum in $c_{exp-sum}$) could be precomputed after each match to optimize the update time. So the practical latency would be simplified and reduced to a \textbf{single} run of \autoref{alg:update}, resulting in a true latency of about 18s before the new rating is obtained.

Similarly to the runtime analysis, the memory costs rise linearly as more matches (opponents) are required (see~\autoref{sfig:memory-user}). Using the \textit{multi-threaded} approach increases memory usage by about~$35\%$, increasing from approximately~8 312 MB (8.3 GB) of memory usage to~11 240 MB (11.2 GB) of memory usage, when $n = 300$ and $\lambda = 128$.

\section{Polynomial Approximation of Equation 2}
\label{app:poly-appro}
We utilised Chebyshev polynomials to approximate \autoref{eq:exp-result} and decided upon the following parameters: 

\begin{table}[h]
\caption{Polynomial approximation details of \autoref{eq:exp-result}}
\centering
\begin{tblr}{
  width = \linewidth,
  colspec = {Q[300]Q[200]},
  row{1} = {azure},
  cell{1}{2} = {c},
  cell{2-4}{2} = {c},
  hlines,
  vlines,
}
                     & \color{white}$\frac{1}{1+10^x}$      \\
Bound                & [-5;+5] \\
Polynomial Degree    & 50        \\
Multiplicative Depth & 7         
\end{tblr}
\label{tbl:poly-approx-details}
\end{table}

We analysed the potential bounds for the function $\frac{1}{1+10^x}$, where $x= (R_{opp}-R_{player})/400$. In this case, the bound will depends on the user rating and equivalently the range in each rank, so if a single rank covers~500 rating points, then the bound will be between the values [-1.25;+1.25], naturally as the gap between ranks increases, the bound increases. To simulate the capabilities of our scheme we choose to use a bound of [-5;+5], which corresponds to a maximum difference of~2000 rating points between players. While this bound is not realistic, it acts as a stress test, showing that the scheme can remain accurate and functional even with large bounds. Following this, we picked the polynomial degree to be~50 as from our experiments this degree showed the best accuracy and performance trade-off. Namely, this choice results in an approximate error of about~$10^{-6}$, from our testing this resulted in the best accuracy for a full update even with bootstrapping, which tends to degrade the precision. We provide the error curve in~\autoref{fig:poly-error}.

\begin{figure}[h]
    \centering
    \includegraphics[width=\linewidth]{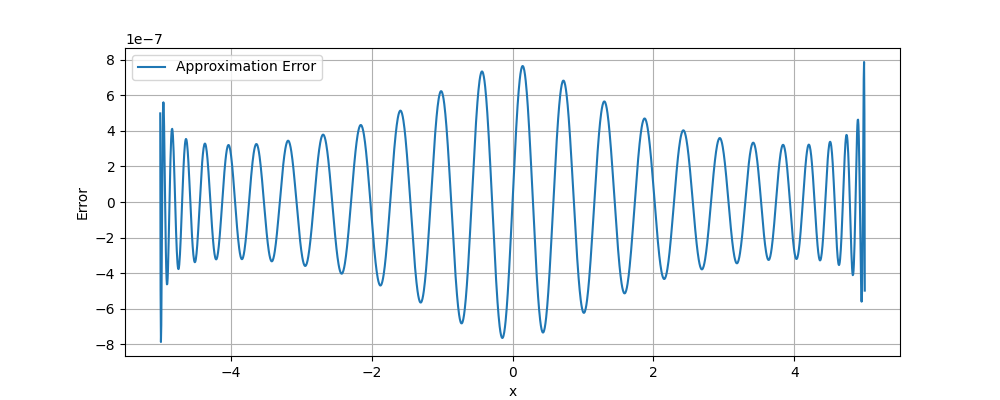}
    \caption{Polynomial Approximation Error, when Bound = [-5;+5] and Degree = 50.}
    \label{fig:poly-error}
\end{figure}

\section{Precision Bits}
\label{app:prec-bits}
Throughout most runs, the estimated precision remained around 52-bits. However, during certain updates it dropped to values between~10-14 bits. We showcase the change of the estimated precision in \autoref{fig:128-precision}. Despite this, the drops did not have a major impact on the actual accuracy results on average as shown in~\autoref{tbl:eval-acc-toy}.

\begin{figure}[h]
    \centering
    \includegraphics[width=\linewidth]{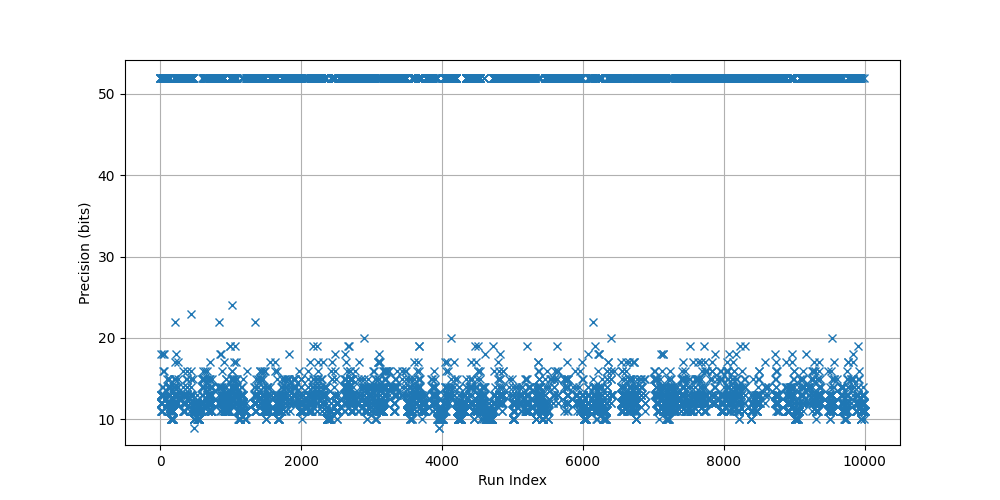}
    \caption{Precision bits after 10 000 consecutive updates, when $\lambda=128$.}
    \label{fig:128-precision}
\end{figure}

\end{document}